\documentclass[a4paper, 11pt]{article}

\pdfoutput=1

\usepackage{graphicx}
\usepackage{color}  
\usepackage{algorithmic}
\usepackage{amssymb}
\usepackage{amsmath}
\usepackage{times}
\usepackage{fullpage}
\ifx\pdftexversion\undefined
  \usepackage[dvips]{graphicx}
\else
  \usepackage{graphicx}
\fi

\newcommand{\punt}[1]{}
\newcommand{\cmnt}[1]{}

\newtheorem{theorem}{Theorem}

\newtheorem{lemma}[theorem]{Lemma}
\newtheorem{corollary}[theorem]{Corollary}

\newtheorem{definition}{Definition}

\newcounter{history}
\newcommand{\hist}[1]{\refstepcounter{history} {#1}}

\newenvironment{proof}[1][Proof]{\noindent\textbf{#1.} }{} %\rule{0.5em}{0.5em}\\}
\newcommand{\secref}[1]{Section~\ref{sec:#1}}
\newcommand{\figref}[1]{Figure~\ref{fig:#1}}

\newcommand{\stref}[1]{step~\ref{step:#1}}

\newcommand{\lemref}[1]{Lemma~\ref{lem:#1}}

\newcommand{\histref}[1]{\ref{hist:#1}}

\newcommand{\theqed}{$\Box$}

\newcommand{\qed}{\hspace*{\fill}\theqed\\\vspace*{-0.5em}}

\newcommand{\Wset}{\textit{Wset}}

\newcommand{\ignore}[1]{}

\newcommand{\tobj} {transaction object}

\newcommand{\txns}[1] {txns(#1)}
\newcommand {\comm}[1] {committed(#1)}
\newcommand {\aborted}[1] {aborted(#1)}
\newcommand {\incomp}[1] {live(#1)}

\newcommand {\tcomp} {t-complete}
\newcommand {\tinc} {t-incomplete}

\newcommand{\evts}[1] {evts(#1)}

\newcommand{\nseq} {non-sequential}
\newcommand{\tseq} {t-sequential}

\newcommand{\lastw} {lastWrite}
\newcommand{\lwrite}[2] {#2.lastWrite(#1)}
\newcommand{\vwrite} {valWrite}
\newcommand{\valw}[2] {#2.valWrite(#1)}

\newcommand{\valid} {valid}
\newcommand{\validity} {validity}
\newcommand{\legal} {legal}
\newcommand{\legality} {legality}

\newcommand{\op} {operation}

\newcommand{\cc} {correctness-criterion}
\newcommand{\ccs} {correctness-criteria}
\newcommand{\gen}[1] {gen(#1)}

\newcommand{\begt} {tbegin}
\newcommand{\tryc} {tryC}
\newcommand{\trya} {tryA}

\newcommand{\rset}[1] {rset(#1)}
\newcommand{\wset}[1] {wset(#1)}

\newcommand{\inv}[1] {#1.inv()}
\newcommand{\rsp}[2] {#1.rsp(#2)}

\newcommand{\opq} {opaque}
\newcommand{\opty} {opacity}
\newcommand{\coop} {co-opaque}
\newcommand{\coopty} {co\text{-}opacity}
\newcommand{\mvop} {mvc-opaque}
\newcommand{\mvopty} {mvc\text{-}opacity}

\newcommand{\vwc} {VWC}
\newcommand{\lo} {LO}
\newcommand{\lopty} {LO}
\newcommand{\lopq} {locally-opaque}

\newcommand{\mvsr} {MVSR}

\newcommand{\csr} {CSR}
\newcommand{\mvs} {multi-version STM}

\newcommand{\perm} {permissive}
\newcommand{\pness} {permissiveness}
\newcommand{\permfn}[1] {perm(#1)}

\newcommand{\ols} {ols-permissive}
\newcommand{\olsness} {ols-permissiveness}
\newcommand{\olsfn}[1] {ols\text{-}perm(#1)}

\newcommand{\multv} {non-single-versioned}
\newcommand{\fmvc} {multi-version conflict order}
\newcommand{\mvc} {mvc}
\newcommand{\mvco} {mvc order}
\newcommand{\mvconflict} {mv-conflict}

\newcommand{\mvlo} {MVLO}

\newcommand{\mvcg}[1] {MVCG(#1)}
\newcommand{\mvg} {multi-version conflict graph}

\title{\bf Multiversion Conflict Notion for Transactional Memory Systems \footnote{A preliminary version of this work was presented at WTTM 2013 and published in \cite{PriSat:MVC:Corr:2013}}}

\author{Priyanka Kumar \\
priyanka@iith.ac.in \\
CSE Department \\
IIT Hyderabad \\
India 
\and 
Sathya Peri \\
sathya\_p@iith.ac.in \\
CSE Department \\
IIT Hyderabad \\
India  
}

\date{}

\begin{document}

\maketitle              % typeset the title of the contribution
\thispagestyle{empty}

\begin{abstract}
In recent years, Software Transactional Memory systems (STMs) have garnered significant interest as an elegant alternative for addressing concurrency issues in memory. STM systems take optimistic approach. Multiple transactions are allowed to execute concurrently. On completion, each transaction is validated and if any inconsistency is observed it is aborted. Otherwise it is allowed to commit.

In databases a class of histories called as conflict-serializability (CSR) based on the notion of conflicts have been identified, whose membership can be efficiently verified. As a result, CSR is the commonly used correctness criterion in databases In fact all known single-version schedulers known for databases are a subset of CSR. Similarly, using the notion of conflicts, a correctness criterion, conflict-opacity (co-opacity) which is a sub-class of can be designed whose membership can be verified in polynomial time. Using the verification mechanism, an efficient STM implementation can be designed that is permissive w.r.t co-opacity. Further, many STM implementations have been developed that using the notion of conflicts.

By storing multiple versions for each transaction object, multi-version STMs provide more concurrency than single-version STMs. But the main drawback of co-opacity is that it does not admit histories that are uses multiple versions. This has motivated us to develop a new conflict notions for multi-version STMs. In this paper, we present a new conflict notion multi-version conflict. Using this conflict notion, we identify a new subclass of opacity, \mvopty{} that admits multi-versioned histories and whose membership can be verified in polynomial time. We show that co-opacity is a proper subset of this class. 

An important requirement that arises while building a multi-version STM system is to decide ``on the spot'' or schedule online among the various versions available, which version should a transaction read from? Unfortunately this notion of online scheduling can sometimes lead to unnecessary aborts of transactions if not done carefully. To capture the notion of online scheduling which avoid unnecessary aborts in STMs, we have identified a new concept \textit{\olsness} and is defined w.r.t a \cc, similar to \pness. We show that it is impossible for a STM system that is permissive w.r.t opacity to such avoid un-necessary aborts i.e. satisfy \olsness{} w.r.t \opty. We show this result is true for \mvopty{} as well. 

%This work introduces the notion of multi-version conflict notion. Using this conflict notion, we define a subclass of \opty, \mvopty{} whose membership can be verified in polynomial time. 
%In fact, we show that it is impossible for a STM system that is permissive w.r.t opacity to such avoid un-necessary aborts. We show this result is true for \mvopty{} as well. 
\end{abstract}

%\keywords
%Concurrency, correctness, atomic operation, opacity, conflict opacity, Multiversion Conflict, software transactional memory, timestamp, multiversion.

\section{Introduction}
\label{sec:intro}

In recent years, Software Transactional Memory systems (STMs) \cite{HerlMoss:1993:SigArch,ShavTou:1995:PODC} have garnered significant interest as an elegant alternative for addressing concurrency issues in memory. STM systems take optimistic approach. Multiple transactions are allowed to execute concurrently. On completion, each transaction is validated and if any inconsistency is observed it is \emph{aborted}. Otherwise it is allowed to \emph{commit}. 

An important requirement of STM systems is to precisely identify the criterion as to when a transaction should be aborted/committed. Commonly accepted \cc{} for STM systems is \emph{opacity} proposed by Guerraoui, and Kapalka \cite{GuerKap:2008:PPoPP}. Opacity requires all the transactions including aborted to appear to execute sequentially in an order that agrees with the order of non-overlapping transactions. Unlike the correctness criterion for traditional databases serializability \cite{Papad:1979:JACM}, opacity ensures that even aborted transactions read consistent values. 

Another important requirement of STM system is to ensure that transactions do not abort unnecessarily. This referred to as the \emph{progress} condition. It would be ideal to abort a transaction only when it does not violate correctness requirement (such as opacity). However it was observed in \cite{attiyaHill:sinmvperm:tcs:2012} that many STM systems developed so far spuriously abort transactions even when not required. A \emph{permissive} STM \cite{Guer+:disc:2008} does not abort a transaction unless committing of it violates the \cc{}.

With the increase in concurrency, more transactions may conflict and abort, especially in presence many long-running transactions which can have a very bad impact on performance \cite{AydAbd:2008:Serial:transact}. Perelman et al \cite{Perel+:2011:SMV:DISC} observe that read-only transactions play a significant role in various types of applications. But long read-only transactions could be aborted multiple times in many of the current STM systems \cite{herlihy+:2003:stm-dynamic:podc,dice:2006:tl2:disc}. In fact Perelman et al \cite{Perel+:2011:SMV:DISC} show that many STM systems waste 80\% their time in aborts due to read-only transactions. 

It was observed that by storing multiple versions of each object, \mvs{s} can ensure that more read \op{s} succeed, i.e., not return abort. History $H\histref{illus}$ shown in \figref{illus} illustrates this idea. \hist{$H\histref{illus}:r_1(x, 0) w_2(x, 10) w_2(y, 10) c_2 r_1(y, 0) c_1$} \label{hist:illus}.
In this history the read on $y$ by $T_1$ returns 0 instead of the previous closest write of 10 by $T_2$. This is possible by having multiple versions for $y$. As a result, this history is \opq{} with the equivalent correct execution being $T_1 T_2$. Had there not been multiple versions, $r_2(y)$ would have been forced to read the only available version which is 10. This value would make the read cause $r_2(y)$ to not be consistent (\opq) and hence abort. 

%\begin{history}
%$r_1(x, 0) w_2(x, 10) w_2(y, 10) c_2 r_1(y, 0) c_1$. 
%\end{history} 
%$H1:r_1(x, 0) w_2(x, 10) w_2(y, 10) c_2 r_1(y, 0) c_1$. 

\begin{figure}[tbph]
%\centerline{\scalebox{0.5}{\input{illus.pdf_t}}}
\centering
\includegraphics[scale=0.7]{./illus.0}
\caption{Pictorial representation of a History $H\histref{illus}$}
\label{fig:illus}
\end{figure}

%It is well known that verifying memberhsip of \vsr{} is NP-Complete \cite{Papad:1979:JACM}. To cope with this hardness \emph{conflict-serializability} or \emph{\csr}, a subclass of \vsr, whose membership can be verified in polynomial time is generally used. 

Checking for membership of \textit{multi-version view-serializability (\mvsr)} \cite[chap. 3]{WeiVoss:2002:Morg}, the correctness criterion for databases, has been proved to be NP-Complete \cite{PapadKanel:1984:MultVer:TDS}. We believe that the membership of opacity, similar to \mvsr, can not be efficiently verified.
%Papad:1979:JACM, Vidya:1987:AINF

In databases a sub-class of \mvsr, \textit{conflict-serializability (\csr)} \cite[chap. 3]{WeiVoss:2002:Morg} has been identified, whose membership can be efficiently verified. As a result, \csr{} is the commonly used correctness criterion in databases since it can be efficiently verified. In fact all known single-version schedulers known for databases are a subset of \csr. Similarly, using the notion of conflicts, a sub-class of opacity, \textit{conflict-opacity (\coopty)} can be designed whose membership can be verified in polynomial time. Further, using the verification mechanism, an efficient STM implementation can be designed that is permissive w.r.t \coopty{} \cite{KuzSat:NI:ICDCN:2014}. Further, many STM implementations have been developed that using the idea of \csr \cite{AydAbd:2008:Serial:transact,SinhaMalik:RuntimeSee:IPDPS:2010}. 

By storing multiple versions for each \tobj, multi-version STMs provide more concurrency than single-version STMs. But the main drawback of \coopty{} is that it does not admit histories that are uses multiple versions. In other words, the set of histories exported by any STM implementation that uses multiple versions is not a subset of \coopty. Thus it can be seen that the \coopty{} does not take advantage of the concurrency provided by using multiple versions. %As a result, it is not clear if a multi-version STM implementation can be developed that is permissive w.r.t some sub-class of opacity. 

This has motivated us to develop a new conflict notions for multi-version STMs. In this paper, we present a new conflict notion \emph{\mvconflict}. Using this conflict notion, we identify a new subclass of \opty, \emph{\mvopty} whose membership can be verified in polynomial time. We further show that \coopty is a proper subset of this class. Further, the conflict notion developed is applicable on \nseq{} histories as well unlike traditional conflicts. 

In this paper, although we employed this conflict notion on \opty{} to develop the sub-class \mvopty, we believe that this conflict notion is generic enough to be applicable on other \cc{} such as local opacity \cite{KuzSat:NI:ICDCN:2014}, virtual worlds consistency \cite{Imbs+:2009:PODC} etc.

An important question that arises while building a multi-version STM system using the proposed \mvconflict{} notion: among the various versions available, which version should a transaction read from? The question was first analyzed in the context of database systems \cite{HadzPapad:AAMVCC:PODS:1985, PapadKanel:1984:MultVer:TDS}. A transactional system (either Database or STM) must decide ``on the spot'' or schedule online which version a transaction can read from based on the past history. 

Unfortunately this notion of online scheduling can sometimes lead to unnecessary aborts of transactions. For instance, suppose a transaction $T_i$ requests a read on \tobj{} $x$. Let the STM system has option of returning a value for $x$ from among two versions, say $v_1$ and $v_2$. Suppose that the STM sytsem returns a version $v_2$. It is possible that this read can cause another $T_j$ to abort in later to maintain correctness. But this abort of $T_1$ could have been avoided if the system returned $v_1$ instead. This concept is better illustrated in \secref{ols} where we show the difficulties with online scheduling. 

To capture the notion of online scheduling which avoid unnecessary aborts in STMs, we have identified a new concept \textit{\olsness}. It is defined w.r.t a \cc, similar to \pness. We show that it is impossible for a STM system that is permissive w.r.t opacity to such avoid un-necessary aborts i.e. satisfy \olsness{} w.r.t \opty. We show this result is true for \mvopty{} as well. We believe that this impossibility result will generalize to other \cc{} such as \lopty{} \cite{KuzSat:NI:ICDCN:2014}. 

\vspace{1mm}
\noindent
\textit{Roadmap.} We describe our system model in \secref{model}. In \secref{mcn} we formally define the conflict notion and describe how to verify its membership in polynomial time using graph characterization. In \secref{ols}, we describe about the difficulty of online scheduling and associated impossibility results. In \secref{disc}, we discuss about extending the \mvconflict{} notion to local-opacity and then give a brief outline of how to develop a STM system using \mvopty. Finally we conclude in \secref{conc}. 

%we describe the working principle of MCN algorithm. In Section~\ref{sec:gar} we are collecting the garbage. Our results are summarized in Section~\ref{sec:conc}.

\section{System Model and Preliminaries}
\label{sec:model}

The notions and definitions described in this section follow the definitions of \cite{KuzSat:NI:ICDCN:2014, Attiya+:DUOp:ICDCS:2013}. We assume a system of $n$ processes (or threads), $p_1,\ldots,p_n$ that access a collection of \emph{objects} via atomic \emph{transactions}. The STM systems is a software library that exports to the processes with the following \emph{transactional operations} or \emph{methods}: (i) \textit{\begt} \op, that starts a new transaction. It returns an unique transaction id; (ii) the \textit{write}$(x,v)$ operation that updates object $x$ with value $v$, (iii) the \textit{read}$(x)$ operation that returns a value read in $x$; (iv) \textit{tryC}$()$ that tries to commit the transaction and returns \textit{ok} or \textit{abort}; (iv) \textit{\trya}$()$ that aborts the transaction and returns \textit{abort}. The objects accessed by the read and write \op{s} are called as \tobj{s}. For the sake of simplicity, we assume that the values written by all the transactions are unique. We also assume that the library ensures \emph{deferred update semantics}, i.e. the write performed by a transaction $T_k$ on a \tobj{} $x$ will be visible to other transactions only after the commit of $T_k$. 

The transactional \op{s} could be non-atomic. To model this, we assume that all these \op{s} have an \emph{invocation} and \emph{response} events. The \op{s} of a transaction consists of the following events: \begt{} consists of $\inv{\begt}$ which is followed by a $\rsp{\begt}{i}$ where $i$ is the id of the transaction. The $read$ by transaction $T_k$ is denoted as $\inv{read_k(x)}$ which is followed by $\rsp{read_k(x)}{v}$ where $v$ is either the current value of $x$ or $A$. Similarly, the $write$ by transaction $T_k$ is denoted as $\inv{write_k(x,v)}$ which is followed by $\rsp{write_k(x,v)}{r}$ where $r$ denotes the result of the write \op. It can either be $ok$ or $A$. The $\tryc$ by transaction $T_k$ is denoted as $\inv{\tryc_k}$ which is followed by $\rsp{\tryc_k}{r}$ where $r$ is either $ok$ or $A$. Similarly, $\trya$ by transaction $T_k$ is denoted as $\inv{\trya_k}$ which is followed by $\rsp{\trya_k}{A}$. When $A$ is returned by an \op, it implies that the transaction $T_k$ is aborted. 

In the case where the \op{s} are atomic, then we simplify the notation. \begt{} is represented as $\begt_k$, read as $read_k(x,v)/read_k(x,A)$, read as $write_k(x,v)/write_k(x,A)$, \tryc{} as $\tryc_k(ok)/\tryc_k(A)$, \trya{} as $\trya_k(A)$. %For simplicity we also refer to $\tryc_k(ok)$ as $c_k$ and a $A$ response to any of the transactional \op{s} of $T_k$ ($read_k(x,A), write_k(x,A), \tryc_k(A), \trya_k(A)$) as $a_k$. 

When the \textit{write}, \textit{read} and \textit{\tryc}$()$ return $A$, we say that the operation is \emph{forcefully aborted}. Otherwise, we say that the operation has \emph{successfully} executed. For simplicity we also refer to $\rsp{\tryc_k}{ok}$ ($\tryc_k(ok)$ in case of atomic \op{s}) as $c_k$. Similarly, when a transactional \op{} returns $A$, i.e. $\rsp{read_k(x)}{A}, \rsp{write_k(x,v)}{A}, \rsp{\tryc_k}{A},\rsp{\trya_k}{A}$ ($read_k(x,A), write_k(x,A), \\
\tryc_k(A), \trya_k(A)$ respectively), we denote the event as $a_k$. Along the same lines, we refer to (non-atomic) read and write \op{s} as $r_k(x,v), w_k(x,v)$ when the invocation and response events are not relevant to the context. Sometimes, we also drop the \tobj{} $x$ and the value $v$ read/written depending on the context. 

For a transaction $T_k$, we denote all the events (\op{s} in case of sequential histories) of $T_k$ as $\evts{T_k}$. All the \tobj{s} read by $T_k$ are denoted as $\rset{T_k}$ and all the \tobj{s} written by it are denoted as $\wset{T_k}$. 

\vspace{1mm}
\noindent
\textit{Histories.} A \emph{history} is a sequence of \emph{events}, i.e., a sequence of invocations and responses of transactional operations. The collection of events is denoted as $\evts{H}$. We denote $<_H$ a total order on the transactional events in $H$. We identify a history $H$ as tuple $\langle \evts{H},<_H \rangle$. \figref{nseq} shows history \hist{$H\histref{nseq}: \inv{w_1(x, 5)}~ \inv{w_2(x, 10)}~ \rsp{w_1(x, 5)}{ok}~ \rsp{w_2(x, 10)}{ok}~ \inv{r_3(x)}~ \rsp{\tryc_1}{ok} \\
\rsp{\tryc_2}{ok}~ \rsp{r_3(x)}{5}$} \label{hist:nseq}. In \figref{nseq}, for simplicity we have not shown inv and rsp events separately. 

We say a history is \emph{sequential} if invocation of each transactional operation is immediately followed by a matching response. For simplicity, we treat each transactional operation as atomic in sequential histories. The order $<_H$ is a total order on the transactional \op{s} in $H$ for sequential histories. History $H\histref{illus}$ shown in \figref{illus} is a sequential history. We also refer to histories which are not sequential as \emph{\nseq}. 

\begin{figure}[tbph]
\centerline{\scalebox{0.5}{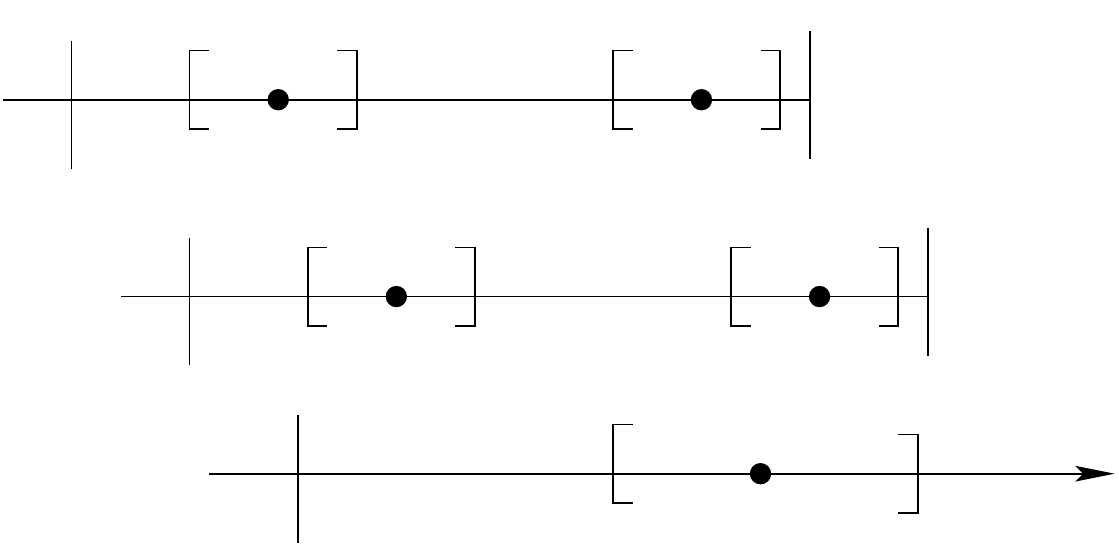}}
\caption{Pictorial representation of a History $H\histref{nseq}$}
\label{fig:nseq}
\end{figure}

Let $H|T_k$ denote the sub-history consisting of events of $T_k$ in $H$. We only consider \emph{well-formed} histories here, i.e., (1) each $H|T_k$ consists of  a read-only prefix (consisting of read operations only), followed by a write-only part (consisting of write operations only), possibly \emph{completed} with a $\tryc$ or $\trya$ operation. In the read-only prefix, each transaction consists of read on a \tobj{} $x$ only once. This restriction brings no loss of  generality \cite{KR:2011:OPODIS}; (2) a thread invoking transactional \op{s} never invokes another operation before receiving a response from the previous one; it does not invoke any operation $op_k$ after receiving a $c_k$ or $a_k$ response.
%$ok$ in response to a $\tryc_k$ \op{} or $A$ in response to any transactional \op{} of $T_k$.  %$c_k$ or $a_k$.

We denote the set of transactions that appear in $H$ is denoted by $\txns{H}$. A transaction $T_k \in \txns{H}$ is \emph{complete} in H if $H|T_k$ ends with a response event. In other words, all the \op{s} in $T_k$ end with a response event. We assume that all the \op{s} in sequential histories are complete. A transaction $T_k \in \txns{H}$ is \emph{\tcomp} if $H|T_k$ ends with $a_k$ or $c_k$ ; otherwise, $T_k$ is \tinc. The history H is \tcomp{} if all transactions in $\txns{H}$ are \tcomp. The set of committed (resp., aborted) transactions in $H$ is denoted by $\comm{H}$ (resp., $\aborted{H}$). The set of \emph{incomplete} or \emph{live} transactions in $H$ is denoted by $\emph{\incomp{H}}$ ($\incomp{H} =\txns{H} - \comm{H} - \aborted{H}$). In $H\histref{nseq}$, $T_3$ is live while $T_1, T_2$ are committed.

We assume that every history has an initial committed transaction $T_0$ that initializes all the \tobj{s} with 0. We say that two histories, $H$ and $H'$ are \emph{equivalent}, denoted as $H \approx H'$ if $\evts{H} = \evts{H'}$ i.e. all the events in $H$ and $H'$ are the same. Note that $H$ could be \nseq{} whereas $H'$ could be sequential. 

\vspace{1mm}
\noindent
\textit{Transaction orders.} For two transactions $T_k,T_m \in \txns{H}$, we say that  $T_k$ \emph{precedes} $T_m$ in the \emph{real-time order} of $H$, denote $T_k \prec_H^{RT} T_m$, if $T_k$ is \tcomp{} in $H$ and the last event of $T_k$ precedes the first event of $T_m$ in $H$. If neither $T_k\prec_H^{RT} T_m$ nor $T_m \prec_H^{RT} T_k$, then $T_k$ and $T_m$ \emph{overlap} in $H$. Consider two histories $h, H'$ that are equivalent to each other, i.e. $\evts{H} = \evts{H'}$. We say a history $H$ \emph{respects} the real-time order of another history $H'$ if all the real-time orders of $H'$ are also in $H$, i.e. $\prec_{H'}^{RT} \subseteq \prec_H^{RT}$.
%$H'$ are also in $H$, i.e. $\prec_{H'}^{RT} \subseteq \prec_H^{RT}$.

A history $H$ is \emph{\tseq{}} if there are no overlapping transactions in $H$, i.e., every two transactions are related by the real-time order. 

\vspace{1mm}
\noindent
\textit{Correctness Criterion.} We denote a collection of histories as \emph{\cc}. Typically, all the histories of a \cc{} satisfy some property. Serializability \cite{Papad:1979:JACM} is the well-accepted \cc{} in databases. Several \ccs{} have been proposed for STMs such as Opacity \cite{GuerKap:2008:PPoPP}, Virtual World Consistency \cite{Imbs+:2009:PODC}, Local Opacity \cite{KuzSat:NI:ICDCN:2014}, TMS \cite{Doherty+:2009:REFINE} etc. 

\vspace{1mm}
\noindent
\textit{Implementations.} A STM \emph{implementation} provides the processes with functions for implementing read, write, \tryc{} (and possibly \trya) functions. We denote the set of histories \emph{generated} by a STM implementation $I$ as $\gen{I}$. We say that an implementation $I$ is correct w.r.t to a \cc{} $C$ if all the histories generated by $I$ are in $C$ i.e. $\gen{I} \subseteq C$. 

\vspace{1mm}
\noindent
\textit{Progress Conditions.} Let $C$ be a \cc{} with $H$ in it. Let $T_a$ be an aborted transaction in $H$. We say that a history $H$ is permissive w.r.t $C$ if committing $T_a$, by replacing the abort value returned by an \op{} in $T_a$ with some non-abort value, would cause $H$ to violate $C$. In other words, if $T_a$ is committed then $H$ will no longer be in $C$. We denote the set of histories permissive w.r.t $C$ as $\permfn{C}$. We say that STM implementation $I$ is permissive \cite{Guer+:disc:2008} w.r.t some \cc{} $C$ (such as opacity) if every history $H$ generated by $I$ is \perm{} w.r.t $C$, i.e., $\gen{I} \subseteq \permfn{C}$.

\section{New Conflict Notion for Multi-Version Systems}
\label{sec:mcn}

In this section, we define a new conflict notion for multi-version STM systems. First, we describe about the \emph{Opacity} \cite{GuerKap:2008:PPoPP}, a popular \cc. Then we describe the new conflict notion, \emph{\fmvc}. 

\subsection{Opacity}
\label{subsec:valid}

We define a few notations on histories for describing opacity. 

\vspace{1mm}
\noindent
\textit{Valid, Legal and Multi-versioned histories.} Let $H$ be a \nseq{} history and $r_k(x, v)$ be a successful read {\op} (i.e $v \neq A$) in $H$. Then $r_k(x, v)$, is said to be \emph{\valid} if there is a transaction $T_j$ in $H$ such that $T_j$ is committed in $H$, $w_j(x, v)$ is in $\evts{T_j}$ and the response of $r_k$ does not occur before invocation of $\tryc_j$ in $H$. Formally, $\langle r_k(x, v)$  is \valid{} $\Rightarrow \exists T_j: (\rsp{r_k(x)}{v} \nless_{H} \inv{\tryc_j}) \land (w_j(x, v) \in \evts{T_j}) \land (v \neq A) \rangle$. We say that the commit \op{} $\rsp{\tryc_j}{ok}$ (or $c_j$) is $r_k$'s \emph{\vwrite} and formally denote it as $\valw{r_k}{H}$. The history $H$ is \valid{} if all its successful read \op{s} are \valid. The notion of \validity{} formalizes deferred update semantics described in \secref{model}.

In $H\histref{nseq}$, $\rsp{\tryc_1}{ok} = c_1 = \valw{r_3(x,5)}{H\histref{nseq}}$, $\rsp{r_k(x)}{5} \nless_{H\histref{nseq}} \inv{\tryc_1}$ and $(w_1(x, 5) \in \evts{T_1})$. Hence, $r_3(x,5)$ is \valid{} and as a result, $H\histref{nseq}$ is \valid{} as well. 

For a sequential history $H$, the definition of \validity{} of $r_k(x,v)$ boils down as follows: a successful read $r_k(x, v)$, is said to be \emph{\valid} if there is a transaction $T_j$ in $H$ that commits before $r_k$ and writes $v$ to $x$. Formally, $\langle r_k(x, v)$  is \valid{} $\Rightarrow \exists T_j: (c_j <_{H} r_k(x, v)) \land (w_j(x, v) \in \evts{T_j}) \land (v \neq A) \rangle$. 

\cmnt{ 
Let $H$ be a history and $r_k(x, v)$ be a successful read {\op} (i.e $v \neq A$) in $H$. Then $r_k(x, v)$, is said to be \emph{\valid} if there is a transaction $T_j$ in $H$ that commits before $r_K$ and $w_j(x, v)$ is in $\evts{T_j}$. Formally, $\langle r_k(x, v)$  is \valid{} $\Rightarrow \exists T_j: (c_j <_{H} r_k(x, v)) \land (w_j(x, v) \in \evts{T_j}) \land (v \neq A) \rangle$. We say that the commit \op{} $c_j$ is $r_k$'s \emph{\vwrite} and formally denote it as $\valw{r_k}{H}$. The history $H$ is \valid{} if all its successful read \op{s} are \valid. 
}

%If there are multiple such committed transactions that write $v$ to $x$, then $r_k$ \vwrite{} is the commit \op{} closest to $r_x$. (non-sequential) 
Consider a sequential history $H$. We define $r_k(x, v)$'s \textit{\lastw{}} as the latest commit event $c_k$ such that $c_k$ precedes $r_k(x, v)$ in $H$ and $x\in\Wset(T_k)$ ($T_k$ can also be $T_0$). Formally, we denote it as $\lwrite{r_k}{H}$. A successful read \op{} $r_k(x, v)$ (i.e $v \neq A$), is said to be \emph{\legal{}} if transaction $T_k$ (which contains  $r_k$'s \lastw{}) also writes $v$ onto $x$. Formally, $\langle r_k(x, v)$ \text{is \legal{}} $\Rightarrow (v \neq A) \land (\lwrite{r_k(x, v)}{H} = c_k) \land (w_k(x,v) \in \evts{T_k}) \rangle$.  The sequential history $H$ is \legal{} if all its successful read \op{s} are \legal. Thus from the definition, we get that if $H$ is \legal{} then it is also \valid.

%\figref{ex1} shows a pictorial representation of a history $H1:r_1(x, 0) w_2(x, 10) w_2(y, 10) c_2 r_1(y, 0) c_1$. 
It can be seen that in $H\histref{illus}$, $c_0 = \valw{r_1(x,0)}{H\histref{illus}} = \lwrite{r_1(x,0)}{H\histref{illus}}$. Hence, $r_1(x,0)$ is \legal. But $c_0 = \valw{r_1(y,0)}{H\histref{illus}} \neq c_1 = \lwrite{r_1(y,0)}{H\histref{illus}}$. Thus, $r_1(y,0)$ is \valid{} but not \legal.

We denote a sequential history $H$ as \textit{\multv} if it is \valid{} but \textbf{not} \legal. If a history $H$ is \multv{}, then there is at least one read, say $r_k(x)$ in $H$ that is \valid{} but not \legal. The history $H\histref{illus}$ is \multv. This definition can not be generalized to \nseq{} histories as \legality{} is not defined for \nseq{} histories. 
%Along the same lines, we say that a STM implementation is \multv{} if it exports atleast one history that is \multv. 

\vspace{1mm}
\noindent
\textit{Opacity.} To define the \cc{} \opty, we first define completion of a history that is incomplete. For a history $H$, we construct the completion of $H$, \emph{opq-completion} denoted $\overline{H^o}$, as follows (similar to \cite{Attiya+:DUOp:ICDCS:2013}):
\begin{enumerate}
\item for every complete transaction $T_k$ in $H$ that is not t-complete, insert the event sequence: \\
$\inv{\trya_k}~ \rsp{\trya_k}{A}$ after the last event of transaction $T_k$; \label{step:ccomp}
\item for every incomplete operation $op_k$ of $T_k$ in $H$, if $op_k = read_k \lor write_k \lor tryA_k$, then insert the response event $A$ somewhere after the invocation of $op_k$;
\item for every incomplete $\tryc_k$ operation where $T_k$ is in $H$, insert response event $ok$ or $A$ somewhere after the invocation of $\tryc_k$. \label{step:camb}
\end{enumerate}

In case of a sequential history $H$, the completion $\overline{H^o}$ is constructed by inserting an $\trya_k(A)$ (or $a_k$) after the last \op{} of transaction $T_k$, for every transaction $T_k$ in $H$ that is \tinc. 

%Next, we define another notion of completion for a sequential history $H$, \emph{conflict-completion} denoted as $\overline{H^c}$. Since the history is sequential all the \op{s} in the history are complete. Thus for any transaction $T_k$ that is incomplete in $H$, we complete it by appending $\trya_k(a)$ \op{} after the last \op{} of transaction $T_k$. This is same as \stref{ccomp} in the definition of completion. 

A history $H$ is said to be \textit{opaque} \cite{GuerKap:2008:PPoPP,tm-book} if $H$ is \valid{} and there exists a \tseq{} legal history $S$ such that (1) $S$ is equivalent to $\overline{H^o}$ and (2) $S$ respects $\prec_{H}^{RT}$, i.e $\prec_{H}^{RT} \subseteq \prec_{S}^{RT}$. 

By requiring $S$ being equivalent to $\overline{H^o}$, opacity treats all the incomplete transactions as aborted. The \validity{} requirement on $H$ ensures that write \op{s} of aborted transactions are ignored. This definition of opacity is closer in spirit to \emph{du-opacity} \cite{Attiya+:DUOp:ICDCS:2013}. It can be seen that both the histories $H\histref{illus}$ and $H\histref{nseq}$ are \opq. The \opq{} equivalent \tseq{} history for $H\histref{illus}$ being $T_1T_2$ and the equivalent \tseq{} histories of $H\histref{nseq}$ are $T_1T_3T_2$, $T_2T_1T_3$ .

\subsection{Motivation for a New Conflict Notion}
\label{subsec:motive}

It is not clear if checking whether a history is opaque or can be performed in polynomial time. Checking for membership of \textit{multi-version view-serializability (\mvsr)} \cite[chap. 3]{WeiVoss:2002:Morg}, the correctness criterion for databases, has been proved to be NP-Complete \cite{PapadKanel:1984:MultVer:TDS}. We believe that the membership of opacity, similar to \mvsr, can not be efficiently verified.
%Papad:1979:JACM, Vidya:1987:AINF

In databases a sub-class of \mvsr, \textit{conflict-serializability (\csr)} \cite[chap. 3]{WeiVoss:2002:Morg} has been identified, whose membership can be efficiently verified. As a result, \csr{} is the commonly used \cc{} in databases since it can be efficiently verified. In fact all known single-version schedulers known for databases are a subset of \csr. Similarly, using the notion of conflicts, a sub-class of opacity, \textit{conflict-opacity (\coopty)} can be designed whose membership can be verified in polynomial time. Further, using the verification mechanism, an efficient STM implementation can be designed that is permissive w.r.t \coopty{} \cite{KuzSat:Corr:2012, KuzSat:NI:ICDCN:2014}.

As already discussed in \secref{intro}, by storing multiple versions for each \tobj, multi-version STMs provide more concurrency than single-version STMs. But the main drawback of \coopty{} is that it does not admit histories that are \multv. Thus \coopty{} does not take advantage of the concurrency provided by using multiple versions. Another big drawback being that \coopty{} does not admit histories that are \nseq. In other words, the set of histories exported by many STM implementation are not a subset of \coopty. Hence, proving correctness of these STM systems is difficult. In the rest of this sub-section, we formally define \coopty{} and show the drawbacks. Some of the definitions and proofs in this section are coming directly from \cite{KuzSat:Corr:2012, KuzSat:NI:ICDCN:2014}.

We define \coopty{} using \textit{conflict order} \cite[Chap. 3]{WeiVoss:2002:Morg}. Consider a sequential history $H$. For two transactions $T_k$ and $T_m$ in $\txns{H}$, we say that \emph{$T_k$  precedes $T_m$ in conflict order}, denoted $T_k \prec_H^{CO} T_m$, (a) (c-c order): $c_k <_H c_m$ and $\wset{T_k} \cap \wset{T_m} \neq\emptyset$; (b) (c-r order): $c_k <_H r_m(x,v)$, $x \in \wset{T_k}$ and $v \neq A$; (c) (r-w order) $r_k(x,v) <_H c_m$, $x \in \wset{T_m}$ and $v \neq A$.

Thus, it can be seen that the conflict order is defined only on \op{s} that have successfully executed. Further, it can also be seen that this order is defined only for histories that are sequential. 

%Next, we define another notion of completion for a sequential history $H$, \emph{conflict-completion} denoted as $\overline{H^c}$. Since the history is sequential all the \op{s} in the history are complete. Thus for any transaction $T_k$ that is incomplete in $H$, we complete it by appending $\trya_k(a)$ \op{} after the last \op{} of transaction $T_k$. This is same as \stref{ccomp} in the definition of completion. 

Using conflict order, \coopty{} is defined as follows: A sequential history $H$ is said to be \emph{conflict opaque} or \emph{\coop} if $H$ is \valid{} and there exists a t-sequential legal history $S$ such that (1) $S$ is equivalent to $\overline{H^o}$ and (2) $S$ respects $\prec_{H}^{RT}$ and $\prec_{H}^{CO}$. 

%Readers familiar with the databases literature may notice \coopty{} is analogous to the \emph{order commit conflict serializability} (OCSR)~\cite{WeiVoss:2002:Morg}.

From the definitions of conflict order and \coopty{} it is clear that these notions are only specific to sequential histories. Thus, history $H\histref{nseq}$ is not \coop. It must be noted that $H\histref{nseq}$ can be generated by a STM system that maintains only a single version of each \tobj. The asynchronous nature of thread execution can result in $H\histref{nseq}$ by the STM system. 

Having seen a drawback, we will next show that if any sequential history is \multv, then it can not be in \coopty. %We can now prove that if a history is \multv, then it is not in \coopty. Due to lack of space, we are only outlining the lemma and theorem statements in this write-up. 

\begin{lemma}
\label{lem:co-eq}
Consider two sequential histories $H1$ and $H2$ such that $H1$ is equivalent to $H2$. Suppose $H1$ respects conflict order of $H2$, i.e., $\prec_{H1}^{CO} \subseteq \prec_{H2}^{CO}$. Then, $\prec_{H1}^{CO} = \prec_{H2}^{CO}$. 
\end{lemma}

\begin{proof}
Here, we have that $\prec_{H1}^{CO} \subseteq \prec_{H2}^{CO}$. In order to prove $\prec_{H1}^{CO} = \prec_{H2}^{CO}$, we have to show that $\prec_{H2}^{CO} \subseteq \prec_{H1}^{CO}$. We prove this using contradiction. Consider two events $p,q$ belonging to transaction $T1,T2$ respectively in $H2$ such that $(p,q) \in \prec_{H2}^{CO}$ but $(p,q) \notin \prec_{H1}^{CO}$. Since the events of $H2$ and $H1$ are same, these events are also in $H1$. This implies that the events $p, q$ are also related by $CO$ in $H1$. Thus, we have that either $(p,q) \in \prec_{H1}^{CO}$  or $(q,p) \in \prec_{H1}^{CO}$. But from our assumption, we get that the former is not possible. Hence, we get that $(q,p) \in \prec_{H1}^{CO} \Rightarrow (q,p) \in \prec_{H2}^{CO}$. But we already have that $(p,q) \in \prec_{H2}^{CO}$. This is a contradiction. \qed
\end{proof}

\begin{lemma}
\label{lem:eqv-legal}
Let $H1$ and $H2$ be two sequential histories which are equivalent to each other and their conflict order are the same, i.e. $\prec_{H1}^{CO} = \prec_{H2}^{CO}$. Then $H1$ is \legal{} iff $H2$ is \legal. 
\end{lemma}

\begin{proof}
It is enough to prove the `if' case, and the `only if' case will follow from symmetry of the argument. Suppose that $H1$ is \legal{}. By contradiction, assume that $H2$ is not \legal, i.e., there is a read \op{} $r_j(x,v)$ (of transaction $T_j$) in $H2$ with its \lastw{} as $c_k$ (of transaction $T_k$) and $T_k$ writes $u \neq v$ to $x$, i.e. $w_k(x, u) \in \evts{T_k}$.  Let $r_j(x,v)$'s \lastw{} in $H1$ be $c_i$ of $T_i$. Since $H1$ is legal, $T_i$ writes $v$ to $x$, i.e. $w_i(x, v) \in \evts{T_i}$. 

Since $\evts{H1} = \evts{H2}$, we get that $c_i$ is also in $H2$, and $c_k$ is also in $H1$.  As $\prec_{H1}^{CO} = \prec_{H2}^{CO}$, we get $c_i <_{H2} r_j(x, v)$ and $c_k <_{H1} r_j(x, v)$. 

Since $c_i$ is the \lastw{} of $r_j(x,v)$ in $H1$ we derive that $c_k <_{H1} c_i$ and, thus, $c_k <_{H2} c_i <_{H2} r_j(x, v)$. But this contradicts the assumption that $c_k$ is the \lastw{} of $r_j(x,v)$ in $H2$. Hence, $H2$ is legal. \qed
\end{proof}

\begin{lemma}
\label{lem:multi-co}
If a sequential history $H$ is \multv{} then $H$ is not in \coopty. Formally, $\langle (H \text{is sequential}) \land (H \text{is \multv}) \implies (H \notin \text{\coopty}) \rangle$. 
\end{lemma}

\begin{proof}
We prove this using contradiction. Assume that $H$ is \multv{}, i.e. $H$ is \valid{} but not \legal. But suppose that $H$ is in \coopty. Since $H$ is sequential, conflict order can be applied on it. From the definition of \coopty, we get that there exists a \tseq{} and legal history $S$ such that $\prec_{H}^{CO} \subseteq \prec_{S}^{CO}$. From \lemref{co-eq}, we get that $\prec_{H}^{CO} = \prec_{S}^{CO}$. Combining this with \lemref{eqv-legal} and the assumption that $H$ is not legal, we get that $S$ is not legal. But this contradicts out assumption that $S$ legal. Hence, $H$ is not in \coopty. \qed
\end{proof}

\subsection{Multi-Version Conflict Definition}
\label{subsec:mvc-defn}

Having seen the shortcomings of \coopty, we will see how to overcome them. The main reason for the shortcoming is because conflict notion has been defined only among the events of sequential histories. We address this issue here by defining a new conflict notion for \nseq{} histories. 

To define this notion on any history, we have developed a another definition of completion of any history $H$, \emph{mvc-completion} denoted as $\overline{H^m}$. It is same as $\overline{H^o}$ except for \stref{camb} which is modified as follows: for every incomplete $\tryc_k$ operation where $T_k$ is in $H$, insert response event $A$ somewhere after the invocation of $\tryc_k$. Thus in $\overline{H^m}$, all incomplete \tryc{} \op{s} are treated as aborted.
%now define a new conflict notion in the next sub-section that will accommodate \multv{} histories as well.

\begin{definition}
\label{defn:mvco}
For a history $H$, we define \textit{\fmvc(\mvco)}, denoted as $\prec^{\mvc}_{H}$, between \op{s} of $\overline{H^{m}}$ as follows: (a) commit-commit (c-c) order: $c_i \prec^{\mvc}_{H} c_j$ if~ $\rsp{\tryc_i}{ok} <_{H} \rsp{\tryc_j}{ok}$ for two committed transaction $T_i$, $T_j$ and both of them write to $x$; (b) commit-read (c-r) order: Let $r_i(x, v)$ be a read \op{} in $H$ with its \vwrite{} as $c_k$ (belonging to the committed transaction $T_k$). Then for any committed transaction $T_j$ that writes to $x$, either the response of the $T_j$'s commit occurs before $T_k$ or $T_k$ is same as $T_j$, formally $(\rsp{\tryc_j}{ok} <_H \rsp{\tryc_k}{ok}) \lor (T_j = T_k)$, we define $c_j \prec^{\mvc}_{H} r_i$. (c) read-commit (r-c) order: Let $r_i(x, v)$ be a read \op{} in $H$ with its \vwrite{} as $c_k$. Then for any committed transaction $T_j$ that writes to $x$, if the $T_j$'s commit response event occurs after $T_k$'s commit response event, i.e. $(\rsp{\tryc_k}{ok} <_H \rsp{\tryc_j}{ok})$, we define $r_i \prec^{\mvc}_{H} c_j$.
\end{definition}

\cmnt {
\begin{definition}
\label{defn:mvco}
For a history $H$, we define \textit{\fmvc(\mvco)}, denoted as $\prec^{\mvc}_{H}$, between \op{s} of $\overline{H_{mvc}}$ as follows: (a) commit-commit (c-c) order: $c_i \prec^{\mvc}_{H} c_j$ if $c_i <_H c_j$ for two committed transaction $T_i$, $T_j$ and both of them write to $x$. (b) commit-read (c-r) order: Let $r_i(x, v)$ be a read \op{} in $H$ with its \vwrite{} as $c_j$ (belonging to the committed transaction $T_j$). Then for any committed transaction $T_k$ that writes to $x$ and either commits before $T_j$ or is same as $T_j$, formally $(c_k <_H c_j) \lor (c_k = c_j)$ , we define $c_k \prec^{\mvc}_{H} r_i$. (c) read-commit (r-c) order: Let $r_i(x, v)$ be a read \op{} in $H$ with its \lastw{} as $c_j$ (belonging to the committed transaction $T_j$). Then for any committed transaction $T_k$ that writes to $x$ and commits after $T_j$, i.e. $c_j <_H c_k$, we define $r_i \prec^{\mvc}_{H} c_k$.
\end{definition}
}

Observe that the \mvco{} is defined on the \op{s} (and not events) of $\overline{H^m}$ and not $H$. The set of conflicts in $H\histref{nseq}$ are: $[\text{c-r}: (c_0, r_3), (c_1, r_3)], [\text{r-c}: (r_3, c_2)], [\text{c-c}: (c_0, c_1), (c_0, c_2), (c_1, c_2)]$. Here, it can be observed that $\rsp{\tryc_2}{ok}$ occurs before $\rsp{r_3(x)}{5}$. Yet, $r_3$ occurs before $c_2$ in the \mvco. 

It is not difficult to extend the \mvco{} to sequential histories: replace the response of a \tryc{} event with the corresponding \tryc{} \op{} and the response of a read event with the corresponding read \op. The set of conflicts in $H\histref{illus}$ are: $[\text{c-r}: (c_0, r_1(x,0)), (c_0, r_1(y))], [\text{r-c}: (r_1(x), c_2), (r_1(y), c_2)], [\text{c-c}: (c_0, c_2)]$. 

We say that a history $H'$ \textit{satisfies} the \mvco{} of a history $H$, $\prec_H^{\mvc}$, denoted as  $H' \vdash \prec_H^{\mvc}$ if: (1) $H'$ is equivalent to $\overline{H^m}$; (2) Consider two \op{s} $op_i, op_j$ in $H$. Let $e_i, e_j$ be the corresponding response events of these \op{s}. Then, $op_i \prec^{\mvc}_{H} op_j$ implies $e_i <_{H'} e_j$. If $H, H'$ are sequential, then $op$ and $e$ would be the same. 

Note that for any sequential history $H$ that is \multv{}, $H$ does not satisfy its own \mvco{} $\prec^{\mvc}_{H}$. For instance the \multv{} order in history $H\histref{illus}$ consists of the pair: $(r_1(y, 0), c_2)$. But $c_2$ occurs before $r_1(y, 0)$ in $H1$. We formally prove this property using the following lemmas.

\cmnt {
\begin{lemma}
\label{lem:satisfy-valid}
Consider a (possibly \multv) sequential and \valid{} history $H$. Let $H'$ be a sequential history which satisfies $\prec^{\mvc}_{H}$. Then $H'$ is \valid{} and $\prec^{\mvc}_{H'} = \prec^{\mvc}_{H}$. Formally, $\langle (H \text{ is \valid}) \land (H, H' \text{ are sequential}) \land (H' \vdash \prec^{\mvc}_{H}) \implies (H' \text{ is \valid}) \land (\prec^{\mvc}_{H'} = \prec^{\mvc}_{H}) \rangle $.
\end{lemma}

\begin{proof}
Here, we have that $H$ is \valid{} and $H'$ satisfies $\prec^{\mvc}_{H}$. Thus all the events of $H$ and $H'$ are the same. The definition of satisfaction says that the events of $H'$ are ordered according to \mvco{} order of $H$. Thus, it can be verified that all the \mvco{s} of both histories are the same, i.e. $\prec^{\mvc}_{H'} = \prec^{\mvc}_{H}$. Since $H$ is \valid{} and $H'$ has the same c-r \mvco{} as $H$, the \vwrite{} of all the read \op{s} in $H'$ occur before the corresponding reads in $H'$. Hence $H'$ is \valid{} as well. \qed
\end{proof}
}

\begin{lemma}
\label{lem:mvco-legal}
Consider a \valid{} history $H$. Let $H'$ be a sequential history (which could be same as $H$). If $H'$ satisfies $\prec^{\mvc}_{H}$ then $H'$ is \legal. Formally, $\langle (H \text{ is \valid}) \land (H' \text{ is sequential}) \land (H' \vdash \prec^{\mvc}_{H}) \implies (H' \text{ is \legal}) \rangle $.
\end{lemma}

\begin{proof}
Assume that $H'$ is not \legal. Hence there exists a read \op{}, say $r_i(x, v)$, in $\evts{H'}$ that is not \legal. This implies that \lastw{} of $r_i$ is not the same as its \vwrite. Let $c_l = \lwrite{r_i}{H'} \neq \valw{r_i}{H'} = c_v$. Let $w_l(x,u) \in \evts{T_l}$ and $w_v(x,v) \in \evts{T_v}$ where $\{T_l, T_v, T_i\} \in \txns{H'}$. As $H$ is \valid, we have that $\rsp{\tryc_v}{ok} <_H \rsp{r_i(x)}{v}$. Since $H' \vdash \prec^{\mvc}_{H}$, we have that $\evts{H} = \evts{H'}$. Thus $\{T_l, T_v, T_i\}$ are also in $\txns{H}$. 

There are two cases w.r.t ordering of events in $H$: %(i) $\rsp{\tryc_l}{ok} <_H \rsp{\tryc_l}{ok}$ (ii) $\rsp{\tryc_l}{ok} <_H \rsp{\tryc_l}{ok}$

\begin{itemize}
\item $\rsp{\tryc_l}{ok} <_H \rsp{\tryc_v}{ok}$: From the definition of \mvco, we get that $\rsp{\tryc_l}{ok} \prec^{\mvc}_H \rsp{\tryc_v}{ok}$. Since $H'$ satisfies $\prec^{\mvc}_H$ and is sequential, we get that $c_l <_{H'} c_v <_{H'} r_i$. 

\item $\rsp{\tryc_v}{ok} <_H \rsp{\tryc_l}{ok}$: Again, from the definition of \mvco, we get that $\rsp{r_i(x)}{v} \\ 
\prec^{\mvc}_H \rsp{\tryc_l}{ok}$. Since $H'$ satisfies $\prec^{\mvc}_H$ and is sequential, we get that $c_v <_{H'} r_i <_{H'} c_l$. 
\end{itemize}

In both cases, it can be seen that $c_l$ is not the previous closest commit \op{} to $r_i$ in $H'$. Hence, we have a contradiction which implies $H'$ is \legal. \qed
\end{proof}

\noindent Using this lemma, we get the following corollary, 

\begin{corollary}
\label{cor:mltv-mvc}
Consider a \valid{} history $H$. Let $H'$ be a \multv{} history equivalent to $H$ (which could be same as $H$). Then, $H'$ does not satisfy $\prec^{\mvc}_{H}$. Formally, $\langle (H' \text{ is \multv}) \land (H \text{ is \valid}) \land (H' \approx H) \implies (H' \nvdash \prec^{\mvc}_{H}) \rangle $.
\end{corollary}

\begin{proof}
We are given that $H$ is \valid, $H$ and $H'$ are equivalent to each other. Since $H'$ is \multv, we get that $H'$ is sequential but not \legal. Combining all these with the contrapositive of \lemref{mvco-legal}, we get that $H' \nvdash \prec^{\mvc}_{H}$. \qed
\end{proof}

\noindent Now, we show that if a history is \legal, then it satisfies it own \mvconflict{} order.

\begin{lemma}
\label{lem:legal-mvco}
Consider a \legal{} history $H$. Then, $H$ satisfies its own \mvconflict{} order $\prec^{\mvc}_{H}$. Formally, $\langle (H \text{ is \legal}) \implies (H \vdash \prec^{\mvc}_{H}) \rangle$.
\end{lemma}

\begin{proof}
We are given that $H$ is \legal. From the definition of \legality, we get that $S$ is sequential. We will prove this lemma using contradiction. Suppose, $H$ does not satisfy its own \mvconflict{} order i.e. $(H \nvdash \prec^{\mvc}_{H})$. Consider two \op{s}, say $p_i$ (belonging to transaction $T_i$) and $q_j$ (belonging to transaction $T_j$) in $\evts{H}$. From our assumption of contradiction, we get that $(p_i \prec_H^{mvc} q_j)$ but $(p_i \nless_H q_j)$. This implies that $(q_j <_H p_i)$ since all the \op{s} are totally ordered in $H$ (which is sequential). Let us consider the various cases of \mvconflict{} between $p_i$ and $q_j$:

\begin{itemize}
\item $p_i=c_i, q_j=c_j$ (c-c order): From \mvconflict{} definition, we get that $c_i \prec^{\mvc}_{H} c_j$ implies that $c_i <_H c_j$. 
\item $p_i=c_i, q_j=r_j$ (c-r order): Let the \vwrite{} of $r_j$ in $H$ be $c_v$ belonging to transaction $T_v$. From \mvconflict{} definition, we get that either $c_i <_H c_v <_H r_j$ or $c_i = c_v <_H r_j$. In either case, we have that $c_i <_H r_j$. 
\item $p_i=r_i, q_j=c_j$ (r-c order): Similar to the above case, Let the \vwrite{} of $r_i$ in $H$ be $c_v$ belonging to transaction $T_v$. From \mvconflict{} definition, we have two option: (i) $c_v <_H c_i <_H r_j$ or (ii) $c_v <_H r_j <_H c_i$. Since $H$ is \legal, option (i) is not possible (unless $c_v = c_i$). This leaves us with option (ii), $r_j <_H c_i$. 
\end{itemize}

Thus in all the three cases, we get that $(p_i <_H q_j)$ which implies that $H$ satisfies $\prec^{\mvc}_{H}$. \qed
\end{proof}

\noindent We now prove an interesting property about satisfaction relation.

\begin{lemma}
\label{lem:satisfy-mvcsubset}
Consider a \valid{} history $H$ and a sequential history $S$. If, $S$ satisfies $H$'s \mvconflict{} order $\prec^{\mvc}_{H}$ then $S$ also respects $H$'s \mvconflict{} order. Formally, $\langle (H \text{ is \valid}) \land (S \text{ is sequential}) \land (S \vdash \prec^{\mvc}_{H}) \implies (\prec^{\mvc}_{H} \subseteq \prec^{\mvc}_{S}) \rangle$.
\end{lemma}

\begin{proof}
We are given that $H$ is \valid, $S$ is sequential and satisfies $H$'s \mvconflict{} order $\prec^{\mvc}_{H}$. Thus, from \lemref{mvco-legal} we get that $S$ is \legal. From \lemref{legal-mvco}, we get that $S$ satisfies its own \mvconflict{} order $\prec^{\mvc}_{H}$, i.e. $S \vdash \prec^{\mvc}_{S}$.

Now, we prove this lemma using contradiction. Suppose, $S$ satisfies $\prec^{\mvc}_{H}$ but $S$ does not respect \mvconflict{} order of $H$, i.e. $\prec^{\mvc}_{H} \nsubseteq \prec^{\mvc}_{S}$. This implies that there exists two \op{s}, $p_i, q_j$ in $H$ and $S$ such that $p_i$ precedes $q_j$ in $H$'s \mvc{} order but not in $S$'s \mvc{} order. We have that, \\
$(p_i \prec^{\mvc}_{H} q_j) \land (p_i \nprec^{\mvc}_{S} q_j) \xrightarrow[\text{satisfy def'n}] {S \vdash \prec^{\mvc}_{S}} (p_i \prec^{\mvc}_{H} q_j) \land (p_i \nless_{S} q_j) \xrightarrow{S \vdash \prec^{\mvc}_{H}} (p_i <_{S} q_j) \land (p_i \nless_{S} q_j)$. This implies a contradiction. Hence, we have that $S$ respects \mvconflict{} order of $H$. \qed

\end{proof}

\subsection{Multi-Version Conflict Opacity}
\label{sec:mvco}

We now illustrate the usefulness of the conflict notion by defining another subset of opacity \emph{\mvopty} which is a superset of \coopty. We formally define it as follows (along the same lines as \coopty): 

\begin{definition}
\label{def:mvcop}
A history $H$ is said to be \emph{multi-version conflict opaque} or \emph{\mvop} if $H$ is \valid{} and there exists a t-sequential history $S$ such that (1) $S$ is equivalent to $\overline{H^m}$, i.e. $S \approx \overline{H^m}$; (2) $S$ respects $\prec_{H}^{RT}$, i.e. $\prec_{H}^{RT} \subseteq \prec_{S}^{RT}$ and $S$ satisfies $\prec_{H}^{\mvc}$, i.e. $S \vdash \prec_{H}^{RT}$. 
\end{definition}

It can be seen that both the histories $H\histref{illus}$ and $H\histref{nseq}$ are \mvop. The \mvc{} equivalent \tseq{} history for $H\histref{illus}$ being $T_1T_2$ and the equivalent \tseq{} history for $H\histref{nseq}$ being $T_1T_3T_2$. 

Consider a history $H$ that is \mvop{} and let $S$ be the \mvc{} equivalent \tseq{} history. Then from \lemref{satisfy-mvcsubset}, we get that $S$ satisfies $H$'s \mvconflict{} order, i.e. $\prec^{\mvc}_{H} \subseteq \prec^{\mvc}_{S}$. Please note that we don't restrict $S$ to be \legal{} in the definition. But it turns out that if $H$ is \mvop{} then $S$ is automatically \legal{} as shown in \lemref{mvco-legal}. Now, we have the following theorem.

\begin{theorem}
\label{thm:mvcop-op}
If a history $H$ is \mvop, then it is also \opq. Formally, $\langle (H \in \mvopty) \implies (H \in \opty) \rangle$.
\end{theorem}

\begin{proof}
Since $H$ is \mvop, it follows that $H$ is \valid{} and there exists a \tseq{} history $S$ such that (1) $S$ is equivalent to $\overline{H^m}$ and (2) $S$ respects $\prec_{H}^{RT}$ and $S$ satisfies $\prec_{H}^{\mvc}$. Since, $S$ is equivalent to $\overline{H^m}$, it can be seen that $S$ is equivalent to $\overline{H^o}$ as well. This, in order to prove that $H$ is \opq, it is sufficient to show that $S$ is \legal. As $S$ satisfies $\prec_{H}^{\mvc}$, from \lemref{mvco-legal} we get that $S$ is \legal. Hence, $H$ is \opq{} as well. \qed 
\end{proof}

Thus, this lemma shows that \mvopty{} is a subset of \opty. Actually, \mvopty{} is a strict subset of \opty. Consider the history \hist{$H\histref{mvcsub} = r_1(x,0) r_2(z,0) r_3(z,0) w_1(x, 5) c_1 r_2(x, 5) w_2(x, 10) w_2(y, 15) \\
c_2 r_3(x, 5) w_3(y, 25) c_3$}. \label{hist:mvcsub}. \figref{nmvc} shows the representation of this history. 
The set of \mvconflict{s} in $H\histref{mvcsub}$ are (ignoring the conflicts with $c_0$): $[\text{c-r}: (c_1, r_2(x,5)), (c_1, r_3(x,5))], [\text{r-c}: (r_3(x,5), c_2)], [\text{c-c}: (c_1, c_2), (c_2, c_3)]$. It can be verified that $H\histref{mvcsub}$ is \opq{} with the equivalent \tseq{} history being $T_1 T_3 T_2$. But there is no \mvc{} equivalent \tseq{} history. This is because of the conflicts: $(r_3(x,5), c_2), (c_2, c_3)$. Hence, $H\histref{mvcsub}$ is not \mvop.

\begin{figure}[tbph]
\centerline{\scalebox{0.5}{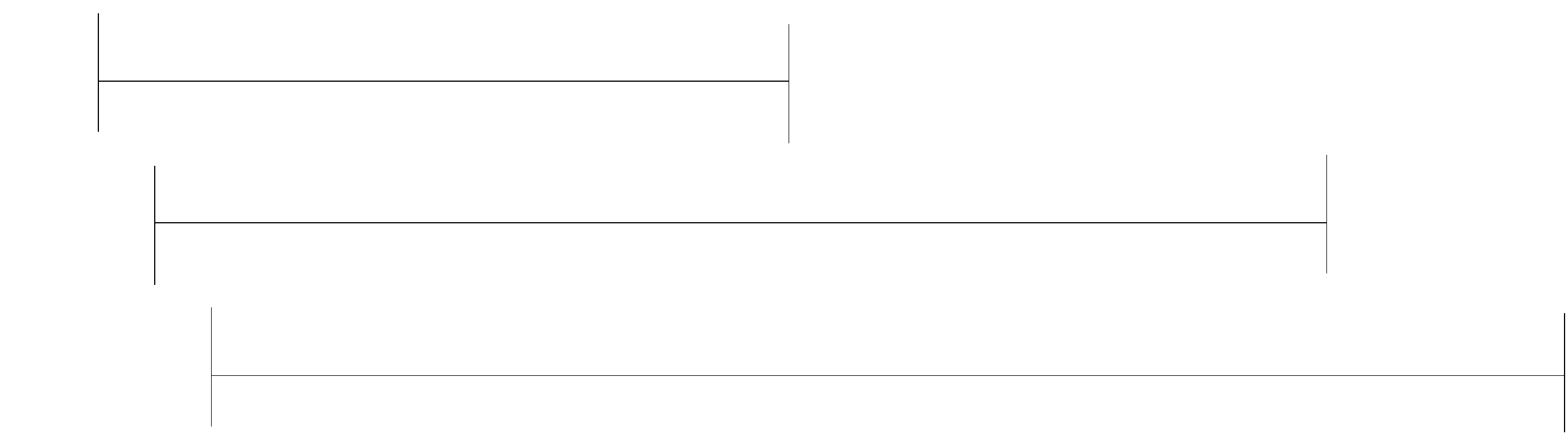}}
\caption{Pictorial representation of $H\histref{mvcsub}$}
\label{fig:nmvc}
\end{figure}

Next, we will relate the classes \coopty{} and \mvopty. In the following theorem, we show that \coopty{} is a subset of \mvopty. 

\begin{theorem}
\label{thm:cop-mvcop}
If a history $H$ is \coop, then it is also \mvop. Formally, $\langle (H \in \coopty) \implies (H \in \mvopty) \rangle$. 
\end{theorem}

\begin{proof}
Since $H$ is \coop, we get that there exists an equivalent \legal{} \tseq{} history $S$ that respects the real-time and conflict orders of $H$. Thus if we show that $S$ satisfies \mvco{} of $H$ then $H$ is \mvop. From the definition of \coopty, we have that $H$ is sequential. 

Since $S$ is \legal, it turns out that the conflicts and \mvconflict{s} are the same. To show this, let us analyse each conflict order:
% It must be noted that since the definition of conflicts is applicable only for sequential histories, we have that $H$ is sequential. 
\begin{itemize}
\item c-c order: If two \op{s} are in c-c conflict, then by definition they are also ordered by the c-c \mvco. 

\item c-r order: Consider the two \op{s}, say $c_k$ and $r_i$ that are in conflict (due to a \tobj{} $x$). Hence, we have that $c_k <_H r_i$. Let $c_v = \valw{r_i}{H}$. Since, $S$ is \legal, either $c_k = c_v$ or $c_k <_H c_j$. In either case, we get that $c_k \prec_H^{\mvc} r_i$.

\item r-c order: Consider the two \op{s}, say $c_k$ and $r_i$ that are in conflict (due to a \tobj{} $x$). Hence, we have that $r_i <_H c_k$. Let $c_v = \valw{r_i}{H}$. Since, $S$ is \legal, $c_v <_H r_i <_H c_k$. Thus in this case also we get that $r_i  \prec_H^{\mvc} c_k$.
\end{itemize}

Thus in all the three cases, conflicts among the \op{s} in $S$ also result in \mvconflict{s} among these \op{s}. Hence, $S$ satisfies the \mvco{} of $H$. \qed
\end{proof}

This theorem shows that \coopty{} is a subset of \mvopty. The history $H\histref{illus}$ is \mvop{} but not in \coop. Hence, \coopty{} is a strict subset of \mvopty. \figref{ex2} shows the relation between the various classes. 

\begin{figure}[tbph]
\centerline{\scalebox{0.5}{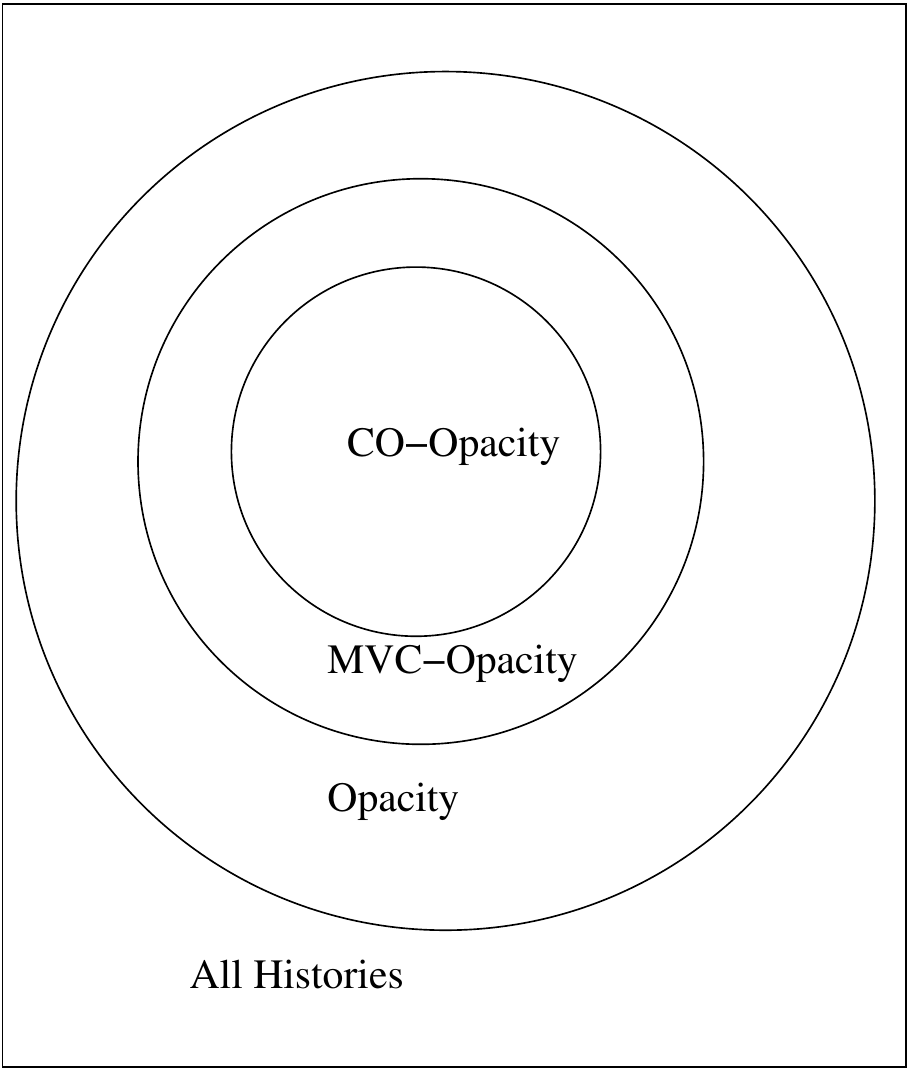}}
\caption{Relation between the various classes}
\label{fig:ex2}
\end{figure}

\subsection{Graph Characterization of MVC-Opacity}
\label{subsec:graph}

In this section, we will describe graph characterization of \mvopty. This characterization will enable us to verify its membership in polynomial time. 

Given a history $H$, we construct a \textit{multi-version conflict graph}, $\mvcg{H} = (V,E)$ as follows:  (1) $V=\txns{H}$, the set of transactions in $H$; (2) an edge $(T_i,T_j)$ is added to $E$ whenever 
\begin{itemize}
\item[2.1] real-time edges: If $T_i$ precedes $T_j$ in $H$;
\item[2.2] \mvco{} edges: If $T_i$ contains an \op{} $p_i$ and $T_j$ contains $p_j$ such that $p_i \prec_{H}^{\mvc} p_j$.
\end{itemize}

\noindent The \mvg{} gives us a polynomial time graph characterization for \mvopty. We show it using the following lemma and theorem.

\begin{lemma}
\label{lem:tseq-graph}
Consider a \legal{} and \tseq{} history $S$. Then, $\mvcg{S}$ is acyclic. Formally, \\
$\langle (S \text{ is \legal}) \land (S \text{ is \tseq})  \implies (\mvcg{S} \text{ is acyclic}) \rangle$.
\end{lemma}

\begin{proof}
Since $S$ is \tseq, we can order all the transactions by their real-time order. We assume w.l.o.g that all the transactions of $S$ are ordered as $T_1 <_S T_2 <_S .... <_S T_n$. Thus, with our assumption we get that $T_i <_S T_j$ implies that $i < j$. 

Now we will show that for any edge $(T_i, T_j)$ in $\mvcg{S}$, we get that $i < j$. The edge $(T_i, T_j)$ can be one of the following:

\begin{itemize}
\item real-time: It follows from this case that $T_j$ started only after the commit of $T_i$. Hence, we get that $T_i <_S T_j$ and this implies $i < j$.

\item c-c conflict: Here, we have that $c_i <_S c_j$. Since $S$ is \tseq, we get that all the events of  $T_i$ occur before all the events of $T_j$. Hence $T_i <_S T_j$ and thus $i < j$.

\item c-r conflict: Here, $c_i <_S r_j$ for a read $r_j(x,v)$. Since $S$ is \tseq, similar to the above case we get that $T_i <_S T_j$ and hence $i < j$.

\item r-c conflict: Here, $r_i <_S c_j$ for a read $r_i(x,v)$. Let \vwrite{} of $r_i$ be $c_l$. From the definition of \mvconflict, we have two cases. Either (i) $c_l <_S c_j <_S r_i$ or (ii) $c_l <_S r_i <_S c_j$. Since $S$ is \legal, we get that case (i) is not possible. Otherwise, $c_j$ would have been the \vwrite{} of $r_i$. This leaves only case (ii) which implies that $r_i <_S c_j$. Since $S$ is \tseq, similar to the above two cases we get that $T_i <_S T_j$ and hence $i < j$.
\end{itemize}

Thus in all the cases, we get that an edge $(T_i, T_j)$ in the $\mvcg{S}$ implies that $i < j$. Hence, a cycle is not possible in such a graph. \qed 
\end{proof}

\cmnt {
\begin{lemma}
\label{lem:hs-satisfy}
Consider a \valid{} history $H$. Let $S$ be a sequential history such that $S$ satisfies $\prec_{H}^{\mvc}$. Then, $S$ satisfies $\prec_{S}^{\mvc}$ as well. Formally, $\langle (H \text{ is \valid}) \land (S \text{ is sequential}) \land (S \vdash \prec_{H}^{\mvc}) \implies (S \vdash \prec_{S}^{\mvc}) \rangle$.
\end{lemma}

\begin{proof}
We prove this lemma using contradiction. Suppose we are given that for a \valid{} history $H$, a sequential history $S$ satisfies $\prec_{H}^{\mvc}$. But $S$ does not satisfy $\prec_{S}^{\mvc}$, i.e. $S \nvdash \prec_{S}^{\mvc}$.

This implies that there are two \op{s} in $S$, $p_i$ belonging to $T_i$ and $q_j$ belonging to $T_j$ which are in \mvconflict. But $q_j$ occurs before $p_i$ in $S$ (note that $S$ is sequential). Now there are three cases, depending on the conflict condition:

\begin{itemize}
\item $p_i = c_i, q_j = c_j$: In this case, we get that $c_j <_S c_i$. But from the definition of the \mvconflict, we get that 
\item $p_i = c_i, q_j = r_j$: 
\item $p_i = r_i, q_j = c_j$: 
\end{itemize}
\qed
\end{proof}
}

\begin{theorem}
\label{thm:graph}
A \valid{} history $H$ is \mvop{} iff $\mvcg{H}$ is acyclic. 
\end{theorem}

\begin{proof} We prove both the directions. \\
\noindent
\vspace{1mm}
\textit{if $\mvcg{H}$ is acyclic then $H$ is \mvop:} Since $\mvcg{H}$ is acyclic, we can perform a topological sort on $\mvcg{H}$. Using the order obtained from the topological sort, we order all the transactions of $\overline{H^m}$ to construct a \tseq{} history $S$. Thus from the construction of $S$, we get that $S$ is equivalent to $\overline{H^m}$. 

It can be seen that $S$ respects $\prec_{H}^{RT}$. If $T_i$ occurs before $T_j$ in $H$, then there is an edge between $T_i$ between $T_j$ in $\mvcg{H}$. This edge ensures that $T_i$ occurs before $T_j$ in $S$ as well.

Consider two \op{s} of $H$, $p_i$ (belonging to $T_i$) and $q_j$ (belonging to $T_j$). If $p_i \prec_{H}^{\mvc} q_j$ then there is an edge between $T_i$ and $T_j$ in $\mvcg{H}$. This edge ensures that $T_i <_S T_j$. Thus, we get that $p_i <_S q_j$. This shows that $S$ satisfies $\prec_{H}^{\mvc}$.

\vspace{1mm}
\noindent
\textit{if $H$ is \mvop{} then $\mvcg{H}$ is acyclic:} Since $H$ is \mvop, we get that there exists a \tseq, \legal{} history $S$ that is equivalent to $H$. We also have that $S$ respects the real-time order of $H$ and satisfies \mvc{} order of $H$. Combining this with \lemref{satisfy-mvcsubset}, we get that $S$ respects the \mvconflict{} order of $H$. Formally, $(\prec_{H}^{RT} \subseteq \prec_{S}^{RT}) \land (\prec_{H}^{\mvc} \subseteq \prec_{S}^{\mvc})$.

Thus, from the graph construction of $\mvcg{H}, \mvcg{S}$, we get that $\mvcg{H} \subseteq \mvcg{S}$. Since $S$ is \legal{} and \tseq, from \lemref{tseq-graph} we get that $\mvcg{S}$ is acyclic. This implies that $\mvcg{H}$ is also acyclic since it is a subgraph of $\mvcg{S}$. 
\qed
\end{proof}

\noindent \figref{mvcgraphs} shows the \mvg{s} for the histories $H\histref{illus}$, $H\histref{nseq}$ and $H\histref{mvcsub}$. In these graphs and other conflict graphs shown in this paper, we have ignored $T_0$ for simplicity.

\begin{figure}[tbph]
\centerline{\scalebox{0.5}{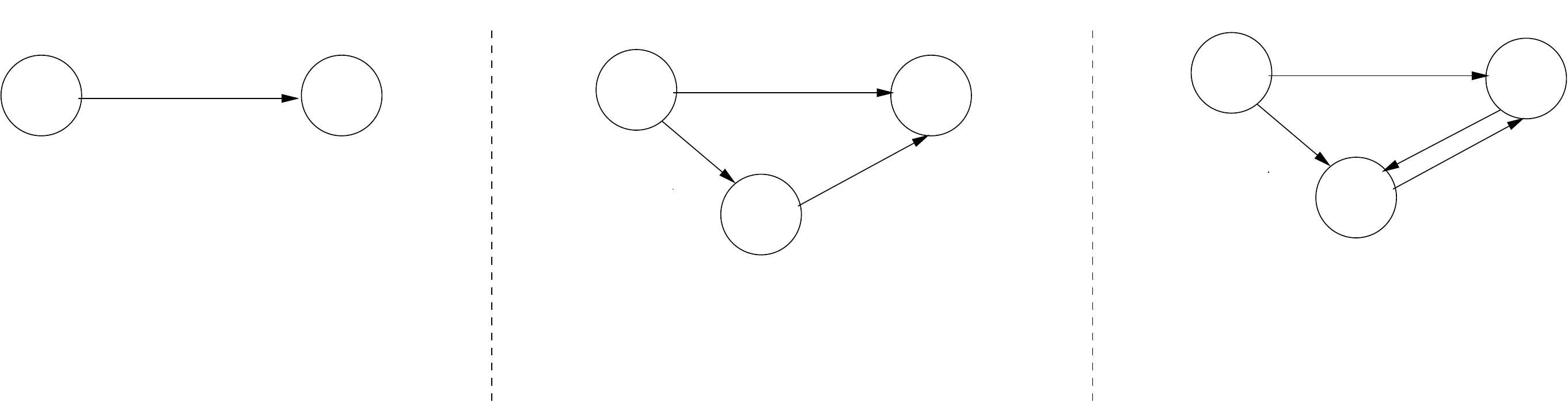}}
\caption{\mvg{s} of $H\histref{illus}$, $H\histref{nseq}$ and $H\histref{mvcsub}$}
\label{fig:mvcgraphs}
\end{figure}

\section{Online Scheduling with Multiple Versions}
\label{sec:ols}

An important question that arises while building a multi-version STM system is among the various versions available, which version should a transaction read from? The question was first analyzed in the context of database systems \cite{HadzPapad:AAMVCC:PODS:1985, PapadKanel:1984:MultVer:TDS}. A transactional system (either Database or STM) must decide ``on the spot'' or schedule online which version a transaction can read from based on the past history. 

We say a STM implementation $I$ \emph{schedules online} (i.e. decides on the spot) if every invocation to an \op{} that it exports (read, write, \tryc, \trya) returns in finite time. We denote $I$ as \emph{online schedulable} (OLS) (term inspired from databases). Note that $I$ can make a decision on scheduling based only on the past history of \op{s} seen so far as it does not have any idea of the future. In other words, all the methods of $I$ are wait-free. 

%This requirement of on the spot scheduling was first explored in databases \cite{HadzPapad:AAMVCC:PODS:1985}. 

%We say a STM implementation $I$ \emph{schedules online} (i.e. decides on the spot) (1) if every invocation to an \op{} that it exports (read, write, \tryc, \trya) returns in finite time (2) the response to an \op{} is based on the \op{s} executed so far, i.e. is based only on the past history. The point of having second condition is to 

But unfortunately this notion of online scheduling can sometimes lead to unnecessary aborts of transactions. We illustrate this idea with an example while considering \mvopty{} as the \cc. Consider the sequential history \hist{$H\histref{ols-ils-pref} = w_1(x, 1) w_1(y, v_1) w_2(x, 2) r_k(z,0) c_1 w_2(z, v_2) c_2 r_3(x, ?^1_2)$}. \label{hist:ols-ils-pref} In this history, $r_3(x)$ has the option of reading $1$ from $T_1$ or $2$ from $T_2$ (denoted as $r_3(x, ?^1_2)$). $T_3$ can not read $x$ from $T_0$ as it would violate the real-time order requirement between $T_0, T_1$ imposed by \mvopty (as well as \opty). Suppose $T_3$ reads $2$ for $x$ written by $T_2$. Now consider a sequence of events that follow the read \op. Let \hist{$H\histref{ols-ils1} = w_1(x, 1) w_1(y, v_1) w_2(x, 2) r_k(z,0) c_1 w_2(z, v_2) r_j(b, 0) \\ 
c_2 r_3(x, 2) w_3(b, v_3) w_k(b, v_k) w_j(d, v_j)$}. \label{hist:ols-ils1}  $H\histref{ols-ils1}$ is a possible extension of $H\histref{ols-ils-pref}$. It can be seen that $H\histref{ols-ils-pref}$ is \mvop{} (with $T_3$ reading $2$). But $H\histref{ols-ils1}$ is not as there is a cycle between the transactions $T_2, T_3, T_k$ in the \mvg. 

Suppose $T_3$ had read $1$ instead of $2$ for $x$. Now consider the modified history consisting of same extension of $H\histref{ols-ils-pref}$ (assuming that the read of $T_3$ did not affect the future events), \hist{$H\histref{ols-ils2} = w_1(x, 1) w_1(y, v_1) w_2(x, 2) r_k(z,0) c_1 w_2(z, v_2) r_j(b, 0) c_2 r_3(x, 1) w_3(b, v_3) w_k(b, v_k) w_j(d, v_j)$}. \label{hist:ols-ils2} It can be seen that $H\histref{ols-ils2}$ is \mvop. $H\histref{ols-ils1}$ will be \mvop{} if $T_k$ is aborted. This shows that the versions read by a transaction can cause other transactions to abort in future. \figref{ols-ils} illustrates this concept. 

\begin{figure}[tbph]
\centerline{\scalebox{0.5}{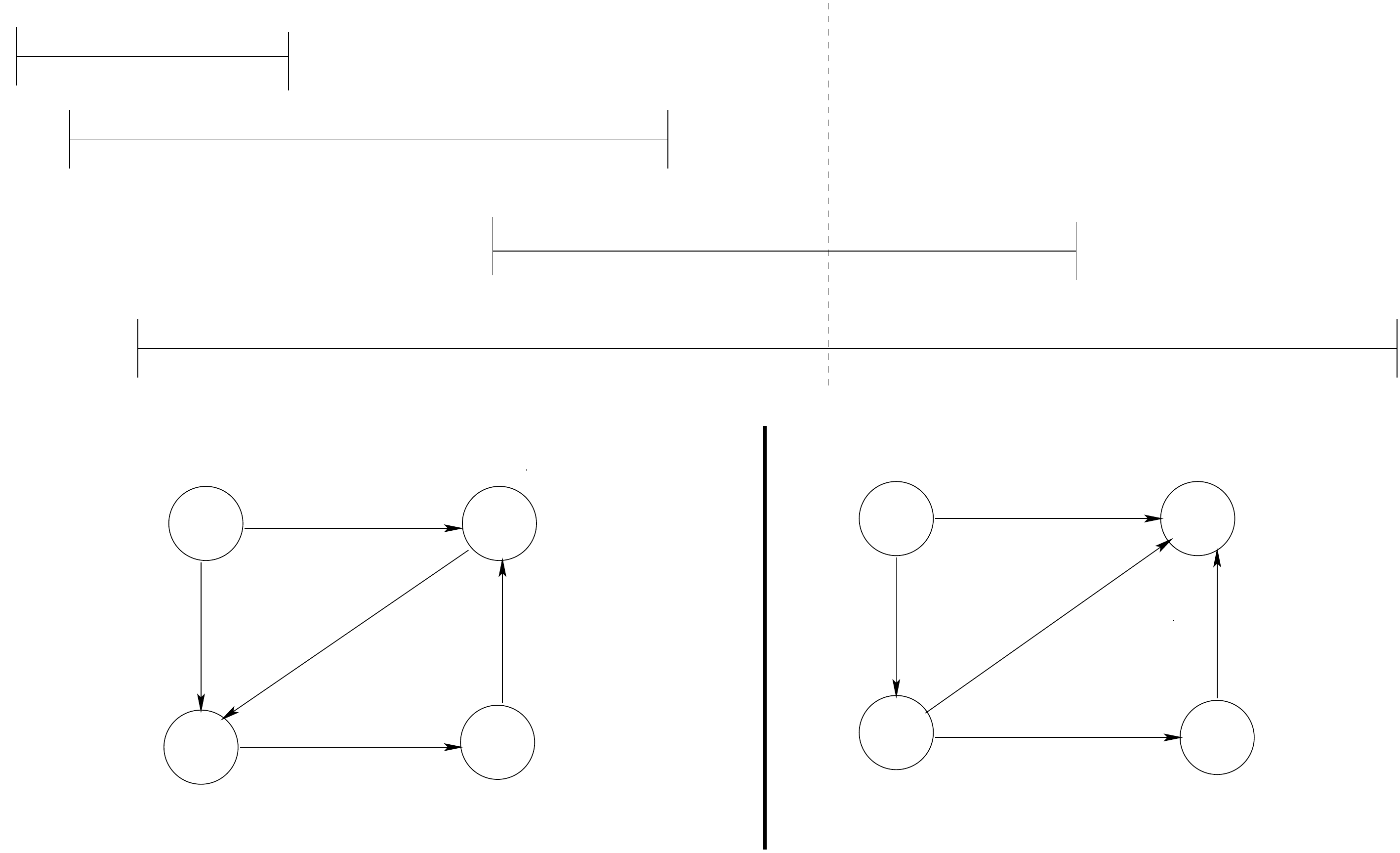}}
\caption{Illustration of difficulties with online scheduling}
\label{fig:ols-ils}
\end{figure}

%The notion of online scheduling in STMs can be combined with \pness. We formally denote this concept as \textit{\olsness} and is defined w.r.t a \cc, similar to \pness. 

To capture the notion of online scheduling which avoid unnecessary aborts in STMs, we have identified a new concept \textit{\olsness} and is defined w.r.t a \cc, similar to \pness. 

Let $C$ be a \cc{} with a history $H$ being permissive w.r.t $C$, i.e. $H \in \permfn{C}$. Then let $T_a$ be an aborted transaction in $H$. Let $r_i(x,v)$ be any successful read \op (i.e. $v \neq A$) in $H$ that completed before the abort response of $T_a$, i.e. $(\rsp{r_i(x)}{v} <_H \rsp{r_a(z)}{A} / \\
\rsp{\tryc_a}{A} / \rsp{\trya_a}{A})$ (for some $r_a$). Suppose $r_i(x)$ read a different value $u$ ($A \neq u \neq v$) from among the various versions available (that were created before by update transactions). Then, committing $T_a$, by replacing the abort value returned by an \op{} in $T_a$ with some non-abort value, would cause $H$ to violate $C$. In other words, if $T_a$ were to be committed with $r_i(x)$ reading $u$, $H$ will no longer be in $C$. We say that $H$ is \emph{\ols} w.r.t $C$. 

In the above example, $H\histref{ols-ils1}$ is not \ols{} w.r.t \mvopty. We denote the set of histories that are \ols{} w.r.t $C$ as $\olsfn{C}$. Along the same lines, we say that STM implementation $I$ is \ols{} w.r.t some \cc{} $C$ (such as opacity) if every history $H$ generated by $I$ is \ols{} w.r.t $C$, i.e., $\gen{I} \subseteq \olsfn{C}$.

It turns out that multiple versions make online scheduling very difficult. In fact we show in the following sub-section that it is impossible to achieve \olsness{}.

\subsection{On Impossibility of \olsness{} with multiple versions} 
\label{subsec:olsness}

As mentioned above, multiple versions make online scheduling very difficult. In this sub-section, we first show that it is impossible for an OLS implementation $I$ that to be \ols{} w.r.t \mvopty. Then, we show that it is impossible for $I$ to be \ols{} w.r.t \opty{} as well.

To show our result, we consider a centralized adversary $\mathcal{A}$ that has complete knowledge of the working of the implementation $I$. We assume that the adversary invokes the next method on the implementation $I$ based on the previous responses. It waits for the response of the previous event before it can fire the next invocation event. Hence, the histories considered in following sub-section are sequential. It must be noted that making this assumption does not restrict the generality of the results as sequential histories are a special case of histories. 

\begin{theorem}
No OLS STM implementation can be \ols{} w.r.t \mvopty. 
\end{theorem}

\begin{proof}
Let us suppose that an OLS STM implementation $I$ is \ols{} w.r.t \mvopty. From the definition of \olsness, we get that $I$ is also \perm{} w.r.t \mvopty. 

Some of the arguments used in this proof are similar to the description in the start of this section. Consider the sequential history \hist{$H\histref{mvc-ols-pref} = w_1(x, 1) w_1(y, v_1) w_2(x, 2) r_k(z,0) c_1 w_2(z, v_2) r_j(b,0) c_2 r_3(x, ?^1_2)$} \label{hist:mvc-ols-pref} (this history is similar to $H\histref{ols-ils-pref}$). Assume that the adversary $\mathcal{A}$ invokes same \op{s} on $I$ as this history. Since $I$ is \perm{} w.r.t \mvopty, it will not unnecessarily return abort to any of these \op{s}. For the read $r_k(z)$, $I$ will return 0 since so far no write to $z$ has taken place. The same argument holds for $r_j(b,0)$. Thus the output by $I$ is same as $H\histref{mvc-ols-pref}$ until $r_3(x)$. 

For $r_3(x)$, $I$ has the option of returning either $1$ or $2$. It can not return $0$ (written by $T_0$) as it violate real-time ordering required by \mvopty. Suppose $I$ returned $2$ for the read $r_3(x)$. Now consider an extension of $H\histref{mvc-ols-pref}$, \hist{$H\histref{mvc-ols1-x2} = w_1(x, 1) w_1(y, v_1) w_2(x, 2) r_k(z,0) c_1 w_2(z, v_2) r_j(b, 0) c_2 r_3(x, 2) w_3(b, v_3) w_k(b, v_k) \\
w_j(d, v_j)$}. \label{hist:mvc-ols1-x2} It can be seen that $H\histref{mvc-ols1-x2}$ is not \mvop{} as there is a cycle between the transactions $T_2, T_3, T_k$ in the \mvg. Suppose $\mathcal{A}$ invokes the \op{s} of $H\histref{mvc-ols1-x2}$ on $I$ after the invocation of $r_3(x)$. Since $H\histref{mvc-ols1-x2}$ is not \mvop, $\mathcal{A}$ invokes the next \op{} only after receiving the previous response and $I$ is \perm{} w.r.t \mvop, $I$ would be forced to abort $T_k$. 

%Since $H\histref{mvc-ols1}$ is not \mvop, $I$ would be forced to abort one of $T_3$, $T_k$. 

Now, consider the case that $I$ had returned $1$ for $r_3(x)$ instead of $2$. The resulting history, \hist{$H\histref{mvc-ols1-x1} = w_1(x, 1) w_1(y, v_1) w_2(x, 2) r_k(z,0) c_1 w_2(z, v_2) r_j(b, 0) c_2 r_3(x, 1) w_3(b, v_3) w_k(b, v_k) w_j(d, v_j) c_j c_k$}. \label{hist:mvc-ols1-x1} It can be seen that $H\histref{mvc-ols1-x1}$ is \mvop{} with an equivalent \tseq{} history being $T_1 T_j T_3 T_k T_2$. Thus, in this case $I$ would not have to abort any transaction. $H\histref{mvc-ols1-x1}$ is in $\olsfn{\mvopty}$. \figref{mvc-ols1} illustrates this scenario. %This suggests that given an option among various available versions, a transaction must read the earlier one. 

\begin{figure}[tbph]
\centerline{\scalebox{0.5}{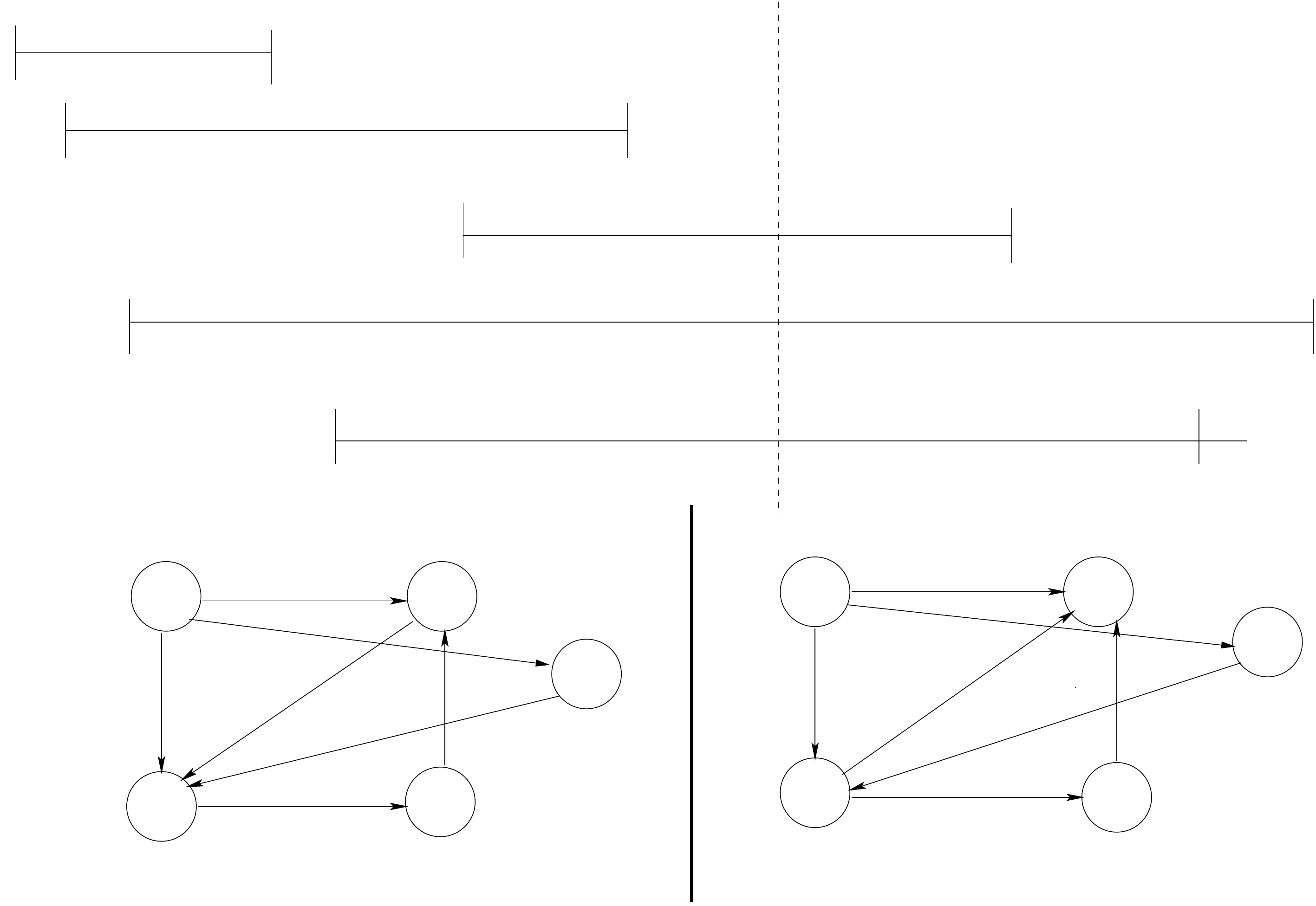}}
\caption{$H\histref{mvc-ols1-x1}$ containing $r_3(x, 1)$ is \mvop}
\label{fig:mvc-ols1}
\end{figure}

Next consider another extension of the history $H\histref{mvc-ols-pref}$, \hist{$H\histref{mvc-ols2-x1} = w_1(x, 1) w_1(y, v_1) w_2(x, 2) r_k(z,0) c_1 \\
w_2(z, v_2) r_j(b, 0) c_2 r_3(x, 1) w_3(b, v_3) w_k(d, v_k) w_j(z, v_j) c_j c_k$}. \label{hist:mvc-ols2-x1} It can be seen that, $H\histref{mvc-ols2-x1}$ is not \mvop{} as there is a cycle between the transactions $T_2, T_j, T_3$ in the \mvg. Suppose $\mathcal{A}$ invokes the \op{s} of $H\histref{mvc-ols2-x1}$ on $I$. Let $I$ returns $1$ for $r_3(x)$ (not knowing what \op{s} could be invoked in future). Then in this case, $I$ would be forced to abort $T_j$ since $H\histref{mvc-ols2-x1}$ is not \mvop{}, $\mathcal{A}$ invokes the next \op{} only after receiving the previous response and $I$ is \perm{} w.r.t \mvop{}. 

On the other hand, suppose $I$ returned $2$ for the above sequence of \op{} invocation by $\mathcal{A}$. The resulting history is \hist{$H\histref{mvc-ols2-x2} = w_1(x, 1) w_1(y, v_1) w_2(x, 2) r_k(z,0) c_1 w_2(z, v_2) r_j(b, 0) c_2 r_3(x, 2) w_3(b, v_3) \\ 
w_k(d, v_k) w_j(z, v_j)$}. \label{hist:mvc-ols2-x2} It can be seen that this history is \mvop{} with an equivalent \tseq{} history being $T_1 T_k T_2 T_j T_3$. Hence, in this case $I$ would output this history without aborting any transaction. $H\histref{mvc-ols2-x2}$ is in $\olsfn{\mvopty}$. \figref{mvc-ols2} illustrates this scenario. 

\begin{figure}[tbph]
\centerline{\scalebox{0.5}{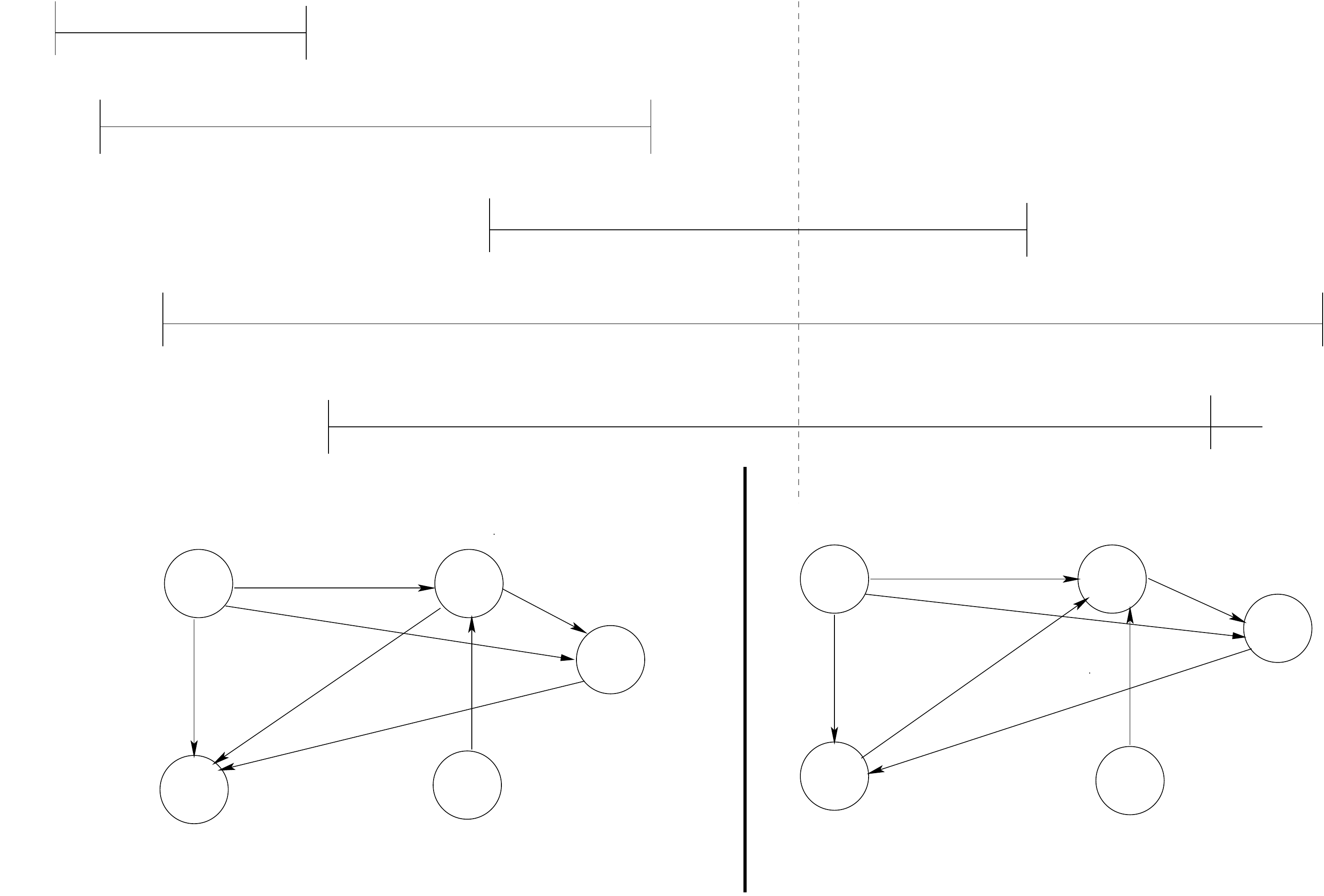}}
\caption{$H\histref{mvc-ols2-x2}$ containing $r_3(x, 2)$ is \mvop}
\label{fig:mvc-ols2}
\end{figure}

These examples illustrate that given the sequence of \op{s} in $H\histref{mvc-ols-pref}$, returning either $1$ or $2$ for $r_3(x)$ by $I$ can possibly cause some transaction in future to abort depending on the sequence of invocations. Whereas reading the other value would have avoided the abort. This is because when $I$ received the event $r_3(x)$, it has no idea about the future events and is OLS. Hence, $I$ can not be in $\olsfn{\mvopty}$. 
\end{proof}

The difficulty of online scheduling is not restricted only to \mvopty. We now show that the impossibility extends to \opty{} as well. In showing this, we use arguments very similar what we have used to the above proof.

\begin{theorem}
No OLS STM implementation can be \ols{} w.r.t \opty. 
\end{theorem}

\cmnt {
\begin{proof}
Let us suppose that an OLS STM implementation $I$ is \ols{} w.r.t \opty. From the definition of \olsness, we get that $I$ is also \perm{} w.r.t \opty. 

Consider the sequential history \hist{$H\histref{opq-ols-pref} = w_0(x, 0) w_0(y, 0) w_0(z, 0) c_0 r_3(x, 0) w_1(x, 1) r_k(z, 0) c_1 r_k(x, ?^0_1)$}. \label{hist:opq-ols-pref} Assume that the adversary $\mathcal{A}$ invokes same \op{s} on $I$ as this history. Since $I$ is \perm{} w.r.t \opty, it will not unnecessarily return abort to any of these \op{s}. For the read $r_3(x)$, $I$ will return 0 since so far no write to $x$ has taken place. The same is true with $r_k(z, 0)$. Thus the output by $I$ is same as $H\histref{opq-ols-pref}$ until $r_k(x)$. 

For $r_k(x)$, $I$ has the option of returning either $0$ or $1$ (the two values written onto $x$ so far). Suppose $I$ returned $1$ for the read $r_k(x)$. Now consider an extension of $H\histref{opq-ols-pref}$, \hist{$H\histref{opq-ols1-x1} = w_0(x, 0) w_0(y, 0) w_0(z, 0) c_0 \\
r_3(x, 0) w_1(x, 1) r_k(z, 0) c_1 r_k(x, 1) r_k(y, 0) c_k w_3(y, 3) c_3$}. \label{hist:opq-ols1-x1} It can be seen that $H\histref{opq-ols1-x1}$ is not \opq. Suppose $\mathcal{A}$ continues with invocation of  the \op{s} of $H\histref{opq-ols1-x1}$ on $I$ after $r_k(x)$. Here $I$ returns $0$ for $r_k(y)$ since no other transaction has written to $y$ yet. Since $H\histref{opq-ols1-x1}$ is not \opq, $\mathcal{A}$ invokes the next \op{} only after receiving the previous response and $I$ is \perm{} w.r.t \opty, $I$ would be forced to abort $T_3$. 

%after the invocation of $r_3(x)$

Now, consider the case that $I$ returns $0$ for $r_k(x)$ instead of $1$. The resulting history, \hist{$H\histref{opq-ols1-x0} = w_0(x, 0) w_0(y, 0) w_0(z, 0) c_0 r_3(x, 0) w_1(x, 1) r_k(z, 0) c_1 r_k(x, 0) r_k(y, 0) c_k w_3(y, 3) c_3$}. \label{hist:opq-ols1-x0} It can be seen that $H\histref{opq-ols1-x0}$ is \opq{} with an equivalent \tseq{} history being $T_0 T_k T_3 T_1$. Thus, in this case $I$ would not have to abort any transaction. $H\histref{opq-ols1-x0}$ is in $\olsfn{\opty}$. \figref{opq-ols1} illustrates this scenario. %This suggests that given an option among various available versions, a transaction must read the earlier one. 

\begin{figure}[tbph]
\centerline{\scalebox{0.5}{\input{opq-ols1.pdf_t}}}
\caption{$H\histref{opq-ols1-x0}$ containing $r_k(x, 0)$ is \opq}
\label{fig:opq-ols1}
\end{figure}

Next consider another extension of the history $H\histref{opq-ols-pref}$, \hist{$H\histref{opq-ols2-x0} = w_0(x, 0) w_0(y, 0) w_0(z, 0) c_0 r_3(x, 0) \\
w_1(x, 1) r_k(z, 0) c_1 r_k(x, 0) w_k(y, k) r_3(y, 0) c_k w_3(y, 3) c_3$}. \label{hist:opq-ols2-x0} It can be seen that, $H\histref{opq-ols2-x0}$ is not \opq. Suppose $\mathcal{A}$ invokes the \op{s} of $H\histref{opq-ols2-x0}$ on $I$. Let $I$ return $0$ for $r_k(x)$ (not knowing what \op{s} could be invoked in future). $I$ returns $0$ for $r_3(y)$ \op{} since no other transaction has yet written to $y$ yet similar to $r_3(x,0)$. Then in this case, $I$ would be forced to abort $T_3$ since $H\histref{opq-ols2-x0}$ is not \opq{}, $\mathcal{A}$ invokes the next \op{} only after receiving the previous response and $I$ is \perm{} w.r.t \opty. 

On the other hand, suppose $I$ returned $1$ for the above sequence of \op{} invocation by $\mathcal{A}$. The resulting history is \hist{$H\histref{opq-ols2-x1} = w_0(x, 0) w_0(y, 0) w_0(z, 0) c_0 r_3(x, 0) w_1(x, 1) r_k(z, 0) c_1 r_k(x, 1) w_k(y, k) r_3(y, 0) \\ 
c_k w_3(y, 3) c_3$}. \label{hist:opq-ols2-x1} It can be seen that this history is \opq{} with an equivalent \tseq{} history being $T_0 T_3 T_1 T_k$. Hence, in this case $I$ would output this history without aborting any transaction. $H\histref{opq-ols2-x1}$ is in $\olsfn{\opty}$. \figref{opq-ols2} illustrates this scenario. 

\begin{figure}[tbph]
\centerline{\scalebox{0.5}{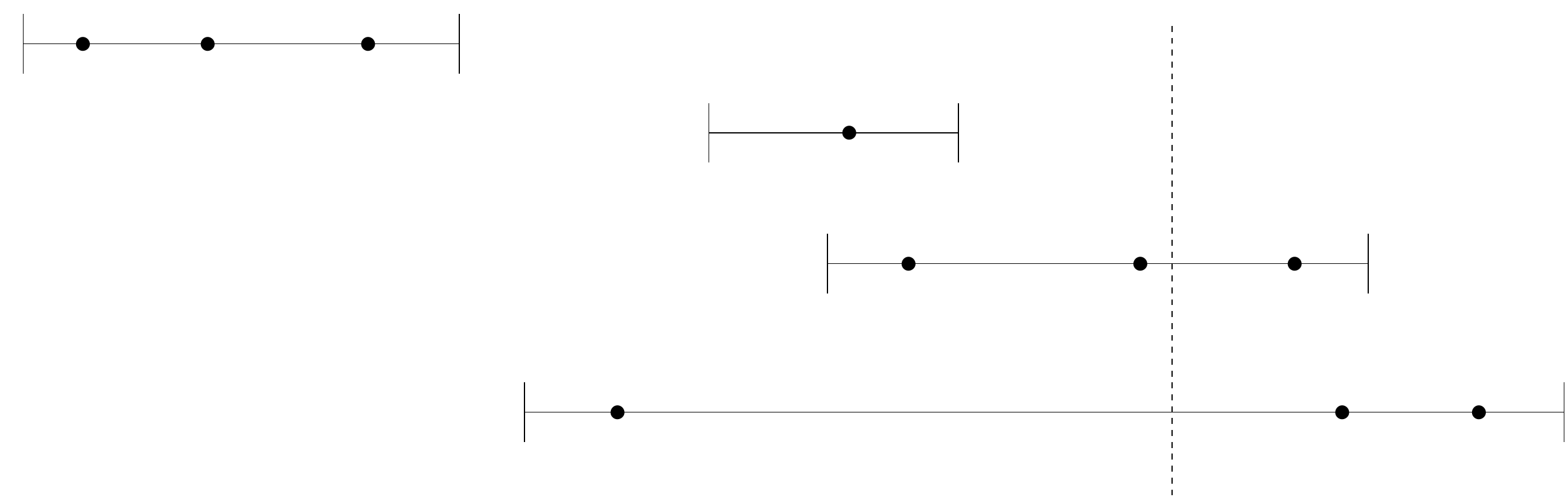}}
\caption{$H\histref{opq-ols2-x1}$ containing $r_k(x, 1)$ is \opq}
\label{fig:opq-ols2}
\end{figure}

These examples illustrate that given the sequence of \op{s} in $H\histref{opq-ols-pref}$, returning either $0$ or $1$ for $r_3(x)$ by $I$ can possibly cause some transaction in future to abort depending on the sequence of invocations. Whereas reading the other value would have avoided the abort. This is because when $I$ received the $r_3(x)$, it has no idea about the future events and is OLS. Hence, $I$ can not be in $\olsfn{\opty}$. \qed
\end{proof}
}

%\input{algo}

%\section{Discussion: Outline of a STM System using Multiversion Conflicts}
\section{Discussion}
\label{sec:disc}

\subsection{Multi-Version Conflicts on other Correctness Criteria}
\label{subsec:othcc}
So far in this paper, we have demonstrated the effectiveness of \mvco{s} using \opty. This conflict notion can be applied to other \cc{} such as local-\opty{} (\lopty) \cite{KuzSat:NI:ICDCN:2014} and virtual world consistency (\vwc) \cite{ImbsRay:2009:SIROCCO}. Both these \ccs{} were defined for sequential histories. %They consider a series of sequential sub-histories. 

A history $H$ is \lopq{} if the following conditions hold: (1) Let the sub-history $H_{com}$ consist of events from all the committed transactions in $H$. Then $H_{com}$ should be \opq; (2) Let $T_a$ be an aborted transaction in $H$. Suppose $H_a$ be a sub-history consisting of all the transactions that committed before the abort of $T_a$ in $H$. Then, for each aborted transaction $T_a$, $H_a$ is \opq. 

% It can be seen that \mvconflict{s} can be applied we can define \emph{multi-version conflict local-opacity} using 
We say a history $H$ is \emph{multi-version conflict local-opaque} (\mvlo) if for each history $H$, (1) $H_{com}$ is \mvop; (2) for each aborted transactions $T_a$, $H_a$ is \mvop. 

Further, it can be seen that the impossibility results of \secref{ols}, can be extended to \mvlo{} and \lo{} as well.

We believe that along the same lines, the multi-version conflict definition can be extended to \vwc. 

\subsection{Outline of a STM System using Multiversion Conflicts}
\label{subsec:algooutline}

Having developed a conflict definition that accommodates multiple versions, we describe the outline of a STM system.The main idea behind the algorithm is based on the notion serialization graph testing \cite{SatVid:2011:ICDCN, KuzSat:NI:ICDCN:2014} that was developed for databases \cite{WeiVoss:2002:Morg}. According to this idea, the STM system maintains a graph based on the \op{s} that have been executed so far. A new \op{} is allowed to execute only if it does not form a cycle in the graph.

But a few important questions arise about the implementation which is typical of any multi-version system: (a) how many version should the STM system store? (b) which version should a transaction read from?

The issue of online scheduling was analyzed in \secref{ols} which partly addresses the question of which version should a transaction read from. Since whichever version a transaction reads from can possibly cause another transaction to abort, in our implementation we have decided to read the closest available version that does not violate \mvop. Using these ideas, we are currently developing a new algorithm. %Some details of the algorithm are described in \cite{PriSat:MVC:Corr:2015}. 

To address the question on number of versions maintained, it was shown in \cite{Kumar+:MVTO:ICDCN:2014} that by not maintaining a limit on the number of versions, greater concurrency can be achieved. So, we do not keep any limit on the number of versions maintained in the STM system developed. But with this approach the number of version keep growing over time making the system inefficient. So, a garbage collection strategy that removes the unwanted versions is to be designed. We are currently working on it.

%\section{Discussion and Conclusion}
\section{Conclusion}
\label{sec:conc}

In this paper, we have presented a new conflict notion \emph{multi-version conflict}. Using this conflict notion, we developed a new subclass of opacity, \mvopty{} that admits multi-versioned histories and whose membership can be verified in polynomial time. We showed that co-opacity, a sub-class of \opty{} that is based on traditional conflicts, is a proper subset of this class. Further, the proposed conflict notion \mvconflict{} can be applied on \nseq{} histories as well unlike traditional conflicts. 

To demonstrate the effectiveness of the new conflict notion, we employed \opty{}, a popular \cc. As discussed, we believe that this conflict notion can be easily extended to other \cc{} such as \lopty{} and \vwc. 

An important requirement that arises while building a multi-version STM system using the propose conflict notion is to decide ``on the spot'' or schedule online among the various versions available, which version should a transaction read from? Unfortunately this notion of online scheduling can sometimes lead to unnecessary aborts of transactions if not done carefully. To capture the notion of online scheduling which avoid unnecessary aborts in STMs, we have identified a new concept \textit{\olsness}. We show that it is impossible for a STM system that is permissive to avoid such un-necessary aborts i.e. satisfy \olsness{} w.r.t \opty. We show this result is true for \mvopty{} as well. 

Actually, multi-version conflict notions have been proposed for multi-version databases as well \cite{HadzPapad:AAMVCC:PODS:1985}. But in their model of histories, the authors do not specify which version a transaction reads. So it is not clear how their model will be applicable to STM histories. Moreover, their notion of conflicts were applicable only for sequential histories. 

As a part of the ongoing work, we plan to develop an efficient STM system using the \mvconflict{s} and measure the cost of the implementation.

{
\bibliography{citations}

\begin{thebibliography}{10}

\bibitem{Attiya+:DUOp:ICDCS:2013}
H.~Attiya, S.~Hans, P.~Kuznetsov, and S.~Ravi.
\newblock Safety of deferred update in transactional memory.
\newblock In {\em Distributed Computing Systems (ICDCS), 2013 IEEE 33rd
  International Conference on}, pages 601--610, July 2013.

\bibitem{attiyaHill:sinmvperm:tcs:2012}
Hagit Attiya and Eshcar Hillel.
\newblock {A {S}ingle-{V}ersion {STM} that is {M}ulti-{V}ersioned
  {P}ermissive}.
\newblock {\em Theory Comput. Syst.}, 51(4):425--446, 2012.

\bibitem{AydAbd:2008:Serial:transact}
Utku Aydonat and Tarek Abdelrahman.
\newblock Serializability of {T}ransactions in {S}oftware {T}ransactional
  {M}emory.
\newblock In {\em TRANSACT~'08: 3rd Workshop on Transactional Computing}, feb
  2008.

\bibitem{dice:2006:tl2:disc}
Dave Dice, Ori Shalev, and Nir Shavit.
\newblock Transactional locking {II}.
\newblock In {\em DISC~'06: Proc. 20th International Symposium on Distributed
  Computing}, pages 194--208, sep 2006.
\newblock Springer-Verlag Lecture Notes in Computer Science volume 4167.

\bibitem{Doherty+:2009:REFINE}
Simon Doherty, Lindsay Groves, Victor Luchangco, and Mark Moir.
\newblock {Towards Formally Specifying and Verifying Transactional Memory}.
\newblock In {\em REFINE}, 2009.

\bibitem{Guer+:disc:2008}
Rachid Guerraoui, Thomas Henzinger, and Vasu Singh.
\newblock Permissiveness in {T}ransactional {M}emories.
\newblock In {\em DISC~'08: Proc. 22nd International Symposium on Distributed
  Computing}, pages 305--319, sep 2008.
\newblock Springer-Verlag Lecture Notes in Computer Science volume 5218.

\bibitem{GuerKap:2008:PPoPP}
Rachid Guerraoui and Michal Kapalka.
\newblock On the {C}orrectness of {T}ransactional {M}emory.
\newblock In {\em PPoPP '08: Proceedings of the 13th ACM SIGPLAN Symposium on
  Principles and practice of parallel programming}, pages 175--184, New York,
  NY, USA, 2008. ACM.

\bibitem{tm-book}
Rachid Guerraoui and Michal Kapalka.
\newblock {\em Principles of Transactional Memory, Synthesis Lectures on
  Distributed Computing Theory}.
\newblock Morgan and Claypool, 2010.

\bibitem{HadzPapad:AAMVCC:PODS:1985}
Thanasis Hadzilacos and Christos~H. Papadimitriou.
\newblock Algorithmic aspects of multiversion concurrency control.
\newblock In {\em Proceedings of the fourth ACM SIGACT-SIGMOD symposium on
  Principles of database systems}, PODS '85, pages 96--104, New York, NY, USA,
  1985. ACM.

\bibitem{HerlMoss:1993:SigArch}
Maurice Herlihy and J.~Eliot B.Moss.
\newblock Transactional memory: {A}rchitectural {S}upport for {L}ock-{F}ree
  {D}ata {S}tructures.
\newblock {\em SIGARCH Comput. Archit. News}, 21(2):289--300, 1993.

\bibitem{herlihy+:2003:stm-dynamic:podc}
Maurice Herlihy, Victor Luchangco, Mark Moir, and III {William N. Scherer}.
\newblock Software transactional memory for dynamic-sized data structures.
\newblock In {\em PODC~'03: Proc. 22nd ACM Symposium on Principles of
  Distributed Computing}, pages 92--101, Jul 2003.

\bibitem{Imbs+:2009:PODC}
Damien Imbs, Jos\'{e}~Ramon de~Mendivil, and Michel Raynal.
\newblock Brief announcement: virtual world consistency: a new condition for
  {STM} systems.
\newblock In {\em PODC '09: Proceedings of the 28th ACM symposium on Principles
  of distributed computing}, pages 280--281, New York, NY, USA, 2009. ACM.

\bibitem{ImbsRay:2009:SIROCCO}
Damien Imbs and Michel Raynal.
\newblock A versatile {STM} protocol with invisible read operations that
  satisfies the virtual world consistency condition.
\newblock In {\em Proceedings of the 16th international conference on
  Structural Information and Communication Complexity}, SIROCCO'09, pages
  266--280, Berlin, Heidelberg, 2010. Springer-Verlag.

\bibitem{PriSat:MVC:Corr:2013}
Priyanka Kumar and Sathya Peri.
\newblock Multi-version conflict notion.
\newblock {\em CoRR}, abs/1307.8256, 2013.

\bibitem{Kumar+:MVTO:ICDCN:2014}
Priyanka Kumar, Sathya Peri, and K.~Vidyasankar.
\newblock A timestamp based multi-version stm algorithm.
\newblock In {\em ICDCN}, pages 212--226, 2014.

\bibitem{KuzSat:Corr:2012}
Petr Kuznetsov and Sathya Peri.
\newblock On non-interference of transactions.
\newblock {\em CoRR}, abs/1211.6315, 2012.

\bibitem{KuzSat:NI:ICDCN:2014}
Petr Kuznetsov and Sathya Peri.
\newblock Non-interference and local correctness in transactional memory.
\newblock In {\em ICDCN}, pages 197--211, 2014.

\bibitem{KR:2011:OPODIS}
Petr Kuznetsov and Srivatsan Ravi.
\newblock On the cost of concurrency in transactional memory.
\newblock In {\em OPODIS}, pages 112--127, 2011.

\bibitem{Papad:1979:JACM}
Christos~H. Papadimitriou.
\newblock The serializability of concurrent database updates.
\newblock {\em J. ACM}, 26(4):631--653, 1979.

\bibitem{PapadKanel:1984:MultVer:TDS}
Christos~H. Papadimitriou and Paris~C. Kanellakis.
\newblock {On Concurrency Control by Multiple Versions}.
\newblock {\em ACM Trans. Database Syst.}, 9(1):89--99, March 1984.

\bibitem{Perel+:2011:SMV:DISC}
Dmitri Perelman, Anton Byshevsky, Oleg Litmanovich, and Idit Keidar.
\newblock {SMV}: {S}elective {M}ulti-{V}ersioning {STM}.
\newblock In {\em DISC}, pages 125--140, 2011.

\bibitem{SatVid:2011:ICDCN}
Sathya Peri and K.Vidyasankar.
\newblock Correctness of concurrent executions of closed nested transactions in
  transactional memory systems.
\newblock In {\em 12th International Conference on Distributed Computing and
  Networking}, pages 95--106, 2011.

\bibitem{ShavTou:1995:PODC}
Nir Shavit and Dan Touitou.
\newblock {S}oftware {T}ransactional {M}emory.
\newblock In {\em PODC '95: Proceedings of the fourteenth annual ACM symposium
  on Principles of distributed computing}, pages 204--213, New York, NY, USA,
  1995. ACM.

\bibitem{SinhaMalik:RuntimeSee:IPDPS:2010}
Arnab Sinha and Sharad Malik.
\newblock Runtime checking of serializability in software transactional memory.
\newblock In {\em IPDPS}, pages 1--12, 2010.

\bibitem{WeiVoss:2002:Morg}
Gerhard Weikum and Gottfried Vossen.
\newblock {\em Transactional Information Systems: Theory, Algorithms, and the
  Practice of Concurrency Control and Recovery}.
\newblock Morgan Kaufmann, 2002.

\end{thebibliography}
}

\end{document}